\UseRawInputEncoding

\documentclass{amsart}
\usepackage{mathrsfs}
\usepackage{graphicx}
\usepackage{amsmath}
\usepackage{amscd}
\usepackage{cite}
\usepackage{hyperref}
\usepackage{epsfig}
\usepackage{subfigure}
\usepackage{caption}
\usepackage{tikz}
  \usepackage{float}
  \usepackage{amssymb}
\usepackage{amsfonts}
\usepackage[all,cmtip]{xy}
\usepackage{appendix}
\usepackage{bm}
\usepackage{diagbox}

\renewcommand{\normalsize}{\fontsize{10pt}{\baselineskip}\selectfont}

\newtheorem{Example}{Example}[section]

\newtheorem{Theorem}{Theorem}[section]
\newtheorem{Theorem/Definition}{Theorem/Definition}[section]
\newtheorem{Proposition}{Proposition}[section]
\newtheorem{Lemma}{Lemma}[section]
\newtheorem{Corollary}{Corollary}[section]

\newcommand{\pd}{\partial}

\newcommand{\bC}{{\mathbb C}}

\newcommand{\bZ}{{\mathbb Z}}

\newcommand{\cB}{{\mathcal B}}

\newcommand{\cF}{{\mathcal F}}
\newcommand{\cG}{{\mathcal G}}

\newcommand{\cM}{{\mathcal M}}

\newcommand{\half}{\frac{1}{2}}

\newcommand{\Mbar}{\overline{\cM}}

\newcommand{\be}{\begin{equation}}
\newcommand{\ee}{\end{equation}}
\newcommand{\bea}{\begin{eqnarray}}
\newcommand{\ben}{\begin{eqnarray*}}
\newcommand{\een}{\end{eqnarray*}}
\newcommand{\eea}{\end{eqnarray}}

\DeclareMathOperator{\Aut}{Aut}

\DeclareMathOperator{\val}{val}

\DeclareMathOperator{\Res}{Res}

\usepackage{color}

\definecolor{yellow}{rgb}{1,1,0}
\definecolor{orange}{rgb}{1,.7,0}
\definecolor{red}{rgb}{1,0,0}
\definecolor{green}{rgb}{0,1,1}
\definecolor{white}{rgb}{1,1,1}

\definecolor{A}{rgb}{.75,1,.75}

\theoremstyle{remark}
\newtheorem{Remark}{Remark}[section]

\begin{document}

\newtheorem{myDef}{Definition}
\newtheorem{thm}{Theorem}
\newtheorem{eqn}{equation}

\title[Topological 1D Gravity and KP Hierarchy]
{Topological 1D Gravity, KP Hierarchy, and\\
Orbifold Euler Characteristics of $\overline{\cM}_{g,n}$}

\author{Zhiyuan Wang}
\address{School of Mathematical Sciences\\
Peking University\\Beijing, 100871, China}
\email{zhiyuan19@math.pku.edu.cn}

\author{Jian Zhou}
\address{Department of Mathematical Sciences\\
Tsinghua University\\Beijing, 100084, China}
\email{jianzhou@mail.tsinghua.edu.cn}

\begin{abstract}

In this work we study the tau-function $Z^{1D}$ of the KP hierarchy specified by the topological 1D gravity.
As an application,
we present two types of algorithms to compute the orbifold Euler characteristics of $\overline\cM_{g,n}$.
The first is to use (fat or thin) topological recursion formulas
emerging from the Virasoro constraints for $Z^{1D}$;
and the second is to use a formula for the connected $n$-point functions
of a KP tau-function in terms of its affine coordinates on the Sato Grassmannian.
This is a sequel to an earlier work \cite{wz2}.

\end{abstract}

\maketitle


\section{Introduction}

The famous Witten Conjecture/Kontsevich Theorem \cite{wi, ko}
claims that a certain generating series of intersection numbers of $\psi$-classes on
the moduli spaces $\Mbar_{g,n}$ of stable curves \cite{dm, kn}
is a solution to the KdV hierarchy.
This striking theorem suggests physicists and mathematicians
to study problems in physics and geometry using techniques in integrable systems.
In this work,
we will apply the techniques in integrable systems
into another geometric problem concerning the moduli spaces $\Mbar_{g,n}$.
We show how to compute the orbifold Euler characteristics of $\Mbar_{g,n}$
using the formalism of emergent geometry of the KP hierarchy.

First let us recall some backgrounds about the orbifold Euler characteristics of $\Mbar_{g,n}$.
It is known in literatures that $\chi(\Mbar_{g,n})$
is a summation over all connected stable graphs of genus $g$ with $n$ external edges
\cite{bh}:
\be
\label{eq-intro-bh}
\chi(\Mbar_{g, n})=n!\cdot
\sum_{\Gamma \in \cG_{g, n}^{c}}
\frac{1}{|\Aut(\Gamma)|} \prod_{v \in V(\Gamma)}
\chi(\cM_{g_{v}, \val_{v}}),
\qquad 2 g-2+n>0,
\ee
where $\cG_{g,n}^c$ denotes the set of all such graphs,
and the orbifold Euler characteristics of $\cM_{g,n}$ is given by the famous Harer-Zagier formula
\cite{hz, pe}:
\be
\label{eq-intro-hz}
\chi(\cM_{g, n})=(-1)^{n} \cdot \frac{(2 g-1) B_{2 g}}{(2 g) !}(2 g+n-3) !, \qquad 2 g-2+n>0.
\ee
It is worthwhile mentioning that in the works \cite{hz, pe},
both Harer-Zagier and Penner have used techniques inspired by physics
(summation over fat graphs, Hermitian matrix models, etc.) to compute $\chi(\cM_{g,n})$.
However,
using \eqref{eq-intro-bh} to compute $\chi(\Mbar_{g, n})$ directly is unpractical
due to the complexity of writing down all the graphs without missing or repeating
any graph.
Furthermore,
the computation of the order of the automorphism group of each graph is also a formidable task.
The number of graphs grows rapidly when $(g,n)$ grows.
The exponential of the generating series of the graph sum formula \eqref{eq-intro-bh}
can be recast into a formal integral formula,
but expanding a formal integral and then taking logarithm is no less harder.
See \cite{ma, ma2} for some results in the case $g=0$,
and \cite{dn} for relating $\chi(\Mbar_{g,n})$ to a problem of counting lattice points in $\Mbar_{g,n}$.

In an earlier work \cite{wz2},
we have derived two types of recursions for $\chi(\Mbar_{g,n})$ which allow us
to compute all these numbers and derive some closed formulas,
using a formalism called the abstract QFT for stable graphs developed in \cite{wz1}.
Such a formalism is inspired by the BCOV holomorphic anomaly equations \cite{bcov, bcov1}.
A byproduct in \cite{wz2} is the following result which relates the computation of $\chi(\Mbar_{g,n})$ to
the topological 1D gravity and the KP hierarchy (see \cite[\S 6]{wz2}):
\begin{Theorem}
[\cite{wz2}]
Let $y,z$ be two formal variable and denote by $\chi(y,z)$ the following generating series
of the orbifold Euler characteristics of $\Mbar_{g,n}$:
\be
\label{eq-intro-gen-chibar}
\chi(y,z):=
\sum_{2g-2+n>0} \frac{y^n z^{2-2g}}{n!} \cdot \chi(\Mbar_{g,n})
-\widetilde V_0(y,z),
\ee
then we have:
\be
\chi(y,z) = F^{1D}(\bm t)|_{t_n = \widetilde V_{n+1}(y,z),n\geq 0} ,
\ee
where $F^{1D}$ is the free energy of the tau-function of the KP hierarchy specified by the topological 1D gravity,
and $\widetilde V_n(y,z)$ are the following formal series:
\be
\label{eq-intro-Vn}
\widetilde V_n(y,z):=
-\half y^2 z^2\cdot\delta_{n,0} + yz\cdot \delta_{n,1} +
\sum_{\substack{g\geq 0,\text{ } g>\half (2-n)}}
\chi(\cM_{g,n})^{2-2g-n}
\ee
for every $n\geq 0$,
and $\chi(\cM_{g,n})$ are given by the Harer-Zagier formula \eqref{eq-intro-hz}.
\end{Theorem}

Then one may expect to use techniques from integrable systems to derive some new algorithms to compute $\chi(\Mbar_{g,n})$.
This is exactly what we do in this work.

The main body of the present paper consists of two parts.
In the first part (\S \ref{sec-1D-KP}-\S \ref{sec-mainresult}),
we focus on the topological 1D gravity and KP hierarchy,
see \cite{ny, zhou6} and \cite{jm, djm, sa} for introductions of these two topics respectively.
The partition function of the topological 1D gravity is the following $1$-dimensional formal integral:
\be
Z^{1D}(\bm t)
:=\frac{1}{(2\pi\lambda^2)^\frac{1}{2}}\int dx
\exp\bigg[\frac{1}{\lambda^2}
\bigg(-\frac{1}{2}x^2+\sum_{n\geq 1}t_{n-1}\frac{x^n}{n!}\bigg)
\bigg],
\ee
involving infinite many coupling constants $\bm t=(t_0,t_1,t_2,\cdots)$.
It is the special case $N=1$ of the Hermitian one-matrix models,
which allows Feynman expansions by both fat and thin graphs (see \cite{zhou9}).

The partition function $Z^{1D}$
is a tau-function of the KP hierarchy (see \cite{ny2}),
and in the first part of this work
we will give full descriptions of this tau-function in Kyoto School's three different pictures:
\begin{itemize}
\item[1)]
On the Sato Grassmannian,
$Z^{1D}$ is described by the subspace $U^{1D}:=\text{span}\{f_n^{1D}(z)\}_{n\geq 0}$ of
$H:=\bC [z]\oplus z^{-1} \bC [[z^{-1}]]$
where (see \S \ref{sec-1d-affinecoord}):
\ben
f_n^{1D}(z)=
\begin{cases}
1 +\sum\limits_{k\geq 0} (2k+1)!! \cdot z^{-2k-2}, & n=0;\\
z^n, & n>0,
\end{cases}
\een
whose coefficients are called the affine coordinates.
Let $Q$ be the following differential operator:
\ben
Q=\pd_z +z -z^{-1},
\een
then
$Q_n:= -z^{n+1}\cdot Q$ is a Kac-Schwarz operator associated to $U^{1D}$ for every $n\geq -1$.
Moreover,
the Virasoro constraints for $Z^{1D}$ are given by these $\{Q_n\}_{n\geq -1}$
via the boson-fermion correspondence
(see \S \ref{sec-KS}-\S \ref{sec-KS-Virasoro}).

\item[2)]
In the fermionic picture,
$Z^{1D}$ is given by (see \S \ref{sec-fermi-1D}):
\begin{equation*}
Z^{1D}=
|0\rangle +\sum_{k=0}^\infty (2k+1)!! \cdot \psi_{-2k-\frac{3}{2}} \psi_{-\half}^* |0\rangle.
\end{equation*}
where $\psi_r$, $\psi_s^*$ are fermionic operators and $|0\rangle$ is the fermionic vacuum.

\item[3)]
In the bosonic picture,
$Z^{1D}$ can be described using Schur functions $s_{\mu}$ or the complete symmetric functions $h_k$ in a simple way
(see \S \ref{sec-fermi-1D}):
\begin{equation*}
Z^{1D} (\bm t)=1+\sum_{k=0} ^\infty (2k+1)!!\cdot s_{(2k+2|0)}
=1+\sum_{k=0} ^\infty (2k+1)!!\cdot h_{2k+2},
\end{equation*}
where the coupling constants $\bm t = (t_0,t_1,\cdots)$ are specialized to the Newton symmetric functions
by $t_n=n!\cdot p_{n+1}$.
\end{itemize}

In the second part (\S \ref{sec-compute-chi}) of this work,
we apply the results on the free energy $\log Z^{1D}$ to
the problem of computations of $\chi(\Mbar_{g,n})$.
Let the specialization $\widetilde\Psi$ be the following linear map:
\begin{equation*}
\begin{split}
\widetilde\Psi :\quad
\prod_{n\geq 1} z_1^{-1}\cdots z_n^{-1} \bC[z_1^{-1},\cdots, z_n^{-1}]
&\quad\to\quad \bC[[y]]((z)),\\
z_1^{-j_1-1}\cdots z_n^{-j_n-1} &\quad\mapsto\quad
\frac{\widetilde V_{j_1}(y,z)\cdots \widetilde V_{j_n}(y,z)}{n!\cdot j_1!\cdots j_n!}
\end{split}
\end{equation*}
where $\widetilde V_j$ are the generating series \eqref{eq-intro-Vn} of $\chi(\cM_{g,n})$,
then we obtain the following algorithms to compute the generating series \eqref{eq-intro-gen-chibar}
of $\chi(\Mbar_{g,n})$
(see \S \ref{sec-compute-chi} for details) :
\begin{Theorem}
We have:
\ben
\chi(y,z) = \widetilde \Psi \big( W^{1D} (\bm z) \big),
\een
where the total $n$-point function $W^{1D}(\bm z)$ can be computed using one of the following three approaches:
\begin{itemize}
\item[1)]
A quadratic recursion formula emerging from the thin Virasoro constraints
(derived in \cite{zhou5}, see \S \ref{eq-qrec-thin} for details);
\item[2)]
An Eynard-Orantin topological recursion emerging from the fat Virasoro constraints
(derived in \cite{zhou5}, see \S \ref{sec-toporec-fat} for details);
\item[3)]
A formula in terms of the affine coordinate of $Z^{1D}$
(see \S \ref{sec-npt-1d} for details).
\end{itemize}
\end{Theorem}

In particular,
the third approach provides a closed formula which does not require recursive computations
(see \S \ref{sec-chi-affine}):
\begin{Theorem}
The generating series $\chi(y,z)$ is given by:
\begin{equation*}
\chi(y,z)
=\widetilde \Psi\bigg(\sum_{n\geq 1} (-1)^{n-1} \sum_{n\text{-cycles }\sigma}
\prod_{i=1}^{n}
\widehat A^{1D} (z_{\sigma(i)}, z_{\sigma(i+1)})
-\frac{1}{(z_1-z_2)^2}
\bigg),
\end{equation*}
where:
\be
\widehat A^{1D}(\xi,\eta)
=\sum_{k=0}^\infty (2k+1)!!\cdot \xi^{-1}\eta^{-2k-2}
+ \frac{1}{\xi-\eta}.
\ee
\end{Theorem}

Here let us briefly summarize all the approaches we have used to the computations of $\chi(\Mbar_{g,n})$
in \cite{wz2} and this work in the following diagram:
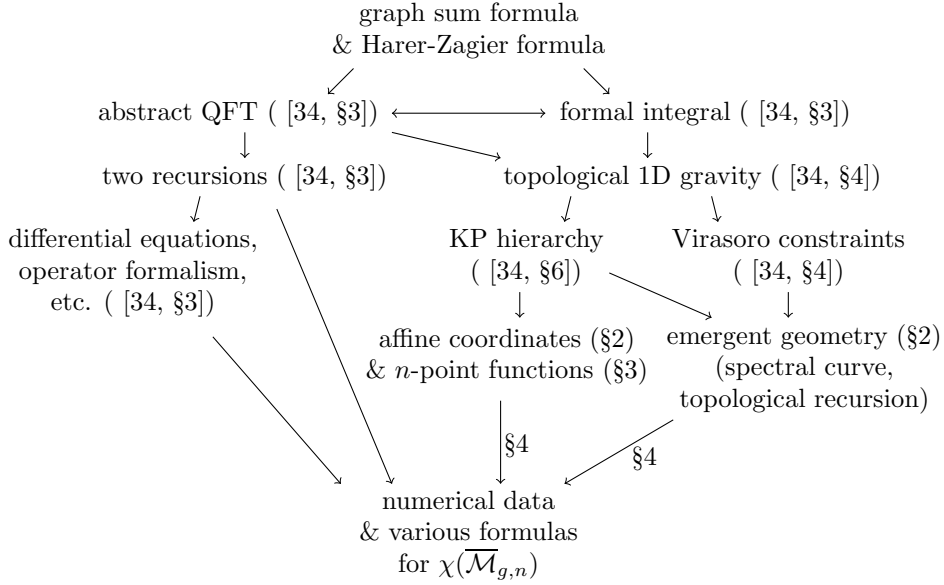
\begin{equation*}
\begin{tikzpicture}[scale=0.85]
\node [align=center,align=center] at (0,0.3) {graph sum formula\\ \& Harer-Zagier formula};
\draw [->] (-1.8,-0.3) -- (-2.2,-0.7);
\draw [->] (1.8,-0.3) -- (2.2,-0.7);
\node [align=center,align=center] at (-3.6,-1) {abstract QFT (\cite[\S 3]{wz2})};
\node [align=center,align=center] at (3.7,-1) {formal integral (\cite[\S 3]{wz2})};
\draw [<->] (-1.2,-1) -- (1.2,-1);
\draw [->] (-3.5,-1.3) -- (-3.5,-1.7);
\draw [->] (-1.2,-1.3) -- (0.5,-1.7);
\node [align=center,align=center] at (-3.5,-2) {two recursions (\cite[\S 3]{wz2})};
\draw [->] (-4.2,-2.3) -- (-4.3,-2.7);
\node [align=center,align=center] at (-5.25,-3.5) {differential equations,\\operator formalism,\\
etc. (\cite[\S 3]{wz2})};
\draw [->] (2.8,-1.3) -- (2.8,-1.7);
\node [align=center,align=center] at (3.5,-2) {topological 1D gravity (\cite[\S 4]{wz2})};
\draw [->] (1.6,-2.3) -- (1.5,-2.7);
\node [align=center,align=center] at (0.9,-3.25) {KP hierarchy\\(\cite[\S 6]{wz2})};
\draw [->] (3.8,-2.3) -- (3.9,-2.7);
\node [align=center,align=center] at (5,-3.25) {Virasoro constraints\\(\cite[\S 4]{wz2})};
\draw [->] (2.2,-3.5) -- (3.8,-4.2);
\draw [->] (5,-3.8) -- (5,-4.2);
\node [align=center,align=center] at (5.25,-5) {emergent geometry (\S \ref{sec-1D-KP})\\(spectral curve,\\topological recursion)};
\draw [->] (0.8,-3.8) -- (0.8,-4.2);
\node [align=center,align=center] at (0.6,-4.8) {affine coordinates (\S \ref{sec-1D-KP})\\ \& $n$-point functions (\S \ref{sec-mainresult})};
\node [align=center,align=center] at (0,-7.6) {numerical data\\ \& various formulas \\for $\chi(\Mbar_{g,n})$};
\draw [->] (-4,-4.5) -- (-2,-6.8);
\draw [->] (-3,-2.5) -- (-1.2,-6.8);
\draw [->] (0.5,-5.5) -- (0.5,-6.8);
\draw [->] (3.2,-5.8) -- (1.5,-6.8);
\node [align=center,align=center] at (0.75,-6.2) {\S \ref{sec-compute-chi}};
\node [align=center,align=center] at (2.75,-6.4) {\S \ref{sec-compute-chi}};
\end{tikzpicture}
\end{equation*}

The rest of this paper is arranged as follows.
In \S \ref{sec-1D-KP} we introduce some basics of the topological 1D gravity.
In \S \ref{sec-mainresult} we describe the tau-function $Z^{1D}$ in Kyoto School's approach
with the help of the affine coordinates
and compute the $n$-point functions.
These results will be applied to the computation of $\chi(\Mbar_{g,n})$ in \S \ref{sec-compute-chi}.
In Appendix \ref{sec-proof-comb},
we show how to calculate the affine coordinates of $Z^{1D}$.
Some numerical data for $Z^{1D}$ and $\chi(\Mbar_{g,n})$ will be presented in Appendix \ref{sec-app-data}.

\section{Topological 1D Gravity}
\label{sec-1D-KP}

In this section,
we recall some basic facts of the topological 1D gravity,
including its genus expansion and Virasoro constraints.
We also recall the emergent geometry for the topological 1D gravity.

\subsection{Topological 1D gravity and genus expansion}
\label{sec-1d-pre}

First let us recall the partition function and free energy of the topological 1D gravity
(see \cite{ny, zhou6}).

The topological 1D gravity is the special $N=1$ case of the Hermitian $N \times N$ one-matrix model.
The partition function $Z^{1D}$ is defined to be the following one-dimensional formal Gaussian integral
(regarded as a formal power series):
\be
\label{eq-1d-partitionZ-t}
Z^{1D}(\bm t)
:=\frac{1}{(2\pi\lambda^2)^\frac{1}{2}}\int dx
\exp\bigg(\frac{1}{\lambda^2}S(x)\bigg),
\ee
where the action functional $S(x)$ is defined by:
\be
\label{eq-1d-actionS-t}
S(x):=-\frac{1}{2}x^2+\sum_{n\geq 1}t_{n-1}\frac{x^n}{n!},
\ee
and $\bm t=(t_0,t_1,t_2,\cdots)$ are the coupling constants.
By directly expanding the formal integral \eqref{eq-1d-partitionZ-t},
one gets the following formula
(\cite[(116)]{zhou6}):
\be
\label{eq-expand-integral}
Z^{1D}(\bm t)= \sum_{n\geq 0}
\sum_{\sum_{j=1}^k j m_j  =2n}
\frac{(2n-1)!!}{\prod_{j=1}^k (j!)^{m_j} m_j!}
\lambda^{2n-2\sum_{j=1}^k m_j}
\cdot \prod_{j=1}^k t_{j-1}^{m_j}.
\ee

The partition function $Z^{1D}$ can be represented
as a summation over (not necessarily stable) ordinary graphs
(see \cite[(95)]{zhou6}):
\be
\label{eq-1d-partitionZ-tFeynman}
Z^{1D} (\bm t) = \sum_{\Gamma \in \cG^{or}} \frac{1}{|\Aut(\Gamma)|}
\prod_{v\in V(\Gamma)} \lambda^{\val(v)-2} t_{\val(v)-1},
\ee
where $\cG^{or}$ is the set of all ordinary graphs
(not necessarily connected,
and not necessarily stable),
$V(\Gamma)$ is the set of vertices of the graph $\Gamma$,
$\Aut(\Gamma)$ is the group of automorphisms of $\Gamma$,
and $\val(v)$ is the valence of the vertex $v\in V(\Gamma)$.
One can use \eqref{eq-expand-integral} or \eqref{eq-1d-partitionZ-tFeynman}
to write down first a few terms of $Z^{1D}$:
\begin{equation*}
\begin{split}
Z^{1D}(\bm t)=& 1 + (\half \lambda^{-2}t_0^2 + \half t_1)
+(\frac{1}{8} \lambda^{-4} t_0^4 +\frac{3}{4} \lambda^{-2} t_0^2 t_1 +\frac{3}{8}t_1^2
+\half t_0t_2 +\frac{1}{8}\lambda^2 t_3)\\
&+(\frac{1}{48}\lambda^{-6}t_0^6 + \frac{5}{16}\lambda^{-4}t_0^4 t_1 + \frac{15}{16}\lambda^{-2}t_0^2 t_1^2 +
\frac{5}{12}\lambda^{-2}t_0^3 t_2 + \frac{5}{16}t_1^3 + \frac{5}{4}t_0t_1t_2 \\
&+\frac{5}{16}t_0^2 t_3 + \frac{1}{8}\lambda^{2}t_0t_4 + \frac{5}{16}\lambda^{2}t_1t_3 +
\frac{5}{24}\lambda^{2}t_2^2 + \frac{1}{48}\lambda^{4}t_5)
+\cdots.
\end{split}
\end{equation*}

\begin{Remark}
A straightforward corollary of the above two formulas
\eqref{eq-expand-integral} and \eqref{eq-1d-partitionZ-tFeynman}
are the following well-known result of graph-counting:
\be
\label{eq-graphcounting}
\frac{(2n-1)!!}{\prod_{j=1}^k (j!)^{m_j}\cdot m_j!}
=\sum_{\Gamma} \frac{1}{|\Aut(\Gamma)|},
\ee
where the summation on the right-hand side is over all (not necessarily connected)
graphs with $m_i$ vertices of valence $i$.
This formula will be useful in Appendix \ref{sec-proof-comb}.
\end{Remark}

By \eqref{eq-1d-partitionZ-tFeynman},
the free energy $F^{1D}(\bm t):= \log Z^{1D}(\bm t)$ has the genus expansion:
\be
\label{eq-genusexp-F}
F^{1D}(\bm t)=\sum_{g=0}^\infty \lambda^{2g-2}F_g^{1D}(\bm t),
\ee
where the free energy $F_g^{1D}$ of genus $g$ is the following summation:
\be
\label{eq-1D-FR-F}
F_g^{1D}(\bm t)
=\sum_{\substack{\Gamma: \text{ connected,}\\ \text{genus}(\Gamma)=g }}\frac{1}{|\Aut(\Gamma)|}
\prod_{v\in V(\Gamma)}t_{\val(v)-1}.
\ee
The `genus' of a connected graph is the number of independent loops in it.
By
\ben
1-\text{genus} = |\{\text{vertices}\}|- |\{\text{edges}\}|
\een
one easily finds that the coefficient of a term $\lambda^{2g-2}t_{a_1}\cdots t_{a_n}$ in $F^{1D}$ is nonzero
only if the following selection rule holds
(\cite[(119)]{zhou6}):
\be
\label{eq-selection-1d}
\sum_{i=1}^n a_i = 2g-2+n.
\ee
Using the above graph sum formula,
one can write down the first few terms of $F^{1D}$:
\begin{equation*}
\begin{split}
F^{1D}(\bm t) =& \lambda^{-2}
(\half t_0^2 + \half t_0^2 t_1 +\half t_0^2t_1^2 +\frac{1}{6}t_0^3 t_2
+\half t_0^2t_1^3+ \half t_0^3t_1t_2+\cdots)\\
&+(\half t_1 + \half t_0 t_2 +\frac{1}{4}t_1^2 +\frac{1}{6}t_1^3+\frac{1}{4}t_0^2t_3
+t_0t_1t_2 +\cdots)\\
&+\lambda^2(\frac{1}{8}t_3+\frac{5}{24}t_2^2 + \frac{1}{8}t_0t_4 +\frac{1}{4}t_1t_3+\cdots)
+\cdots\cdots.
\end{split}
\end{equation*}

\begin{Remark}
Unfortunately,
neither taking the logarithm of \eqref{eq-expand-integral} nor using the Feynman graph expansion \eqref{eq-1D-FR-F}
is efficient for the computation of the coefficients of $\lambda^{2g-2}t_{a_1}\cdots t_{a_n}$ in $F^{1D}(\bm t)$
when $g$ or $a_1,\cdots,a_n$ are large.
In \cite[\S 7]{wz1},
we have developed a simple quadratic recursion to compute $F_g^{1D}$ in the renormalized coordinates $\{I_k\}$
(see \cite{iz, zhou6}).
The correlators of $F^{1D}$ in the original coupling constants $\{t_k\}$
can be solved recursively using the Virasoro constraints or
the topological recursion emerging from the Virasoro constraints
\cite{zhou5}.
In this work we will give a formula for the connected $n$-point functions
using affine coordinates in \S \ref{sec-mainresult}.
\end{Remark}

\subsection{Flow equations, polymer equation, and Virasoro constraints}
\label{sec-pre-Virasoro}

Now we recall the flow equations, polymer equation,
and Virasoro constraints for $Z^{1D}$.

First by taking partial derivatives directly to
the formal Gaussian integral \eqref{eq-1d-partitionZ-t},
one is able to derive the flow equations and polymer equation for $Z^{1D}$:
\begin{Theorem}
[\cite{ny, zhou6}]
\label{thm-1D-flow&polymer}
The partition function $Z^{1D}$ satisfies the flow equations:
\be
\label{eq-1d-flow}
\frac{\pd}{\pd t_n} Z^{1D}=
\frac{\lambda^{2n}}{(n+1)!}\frac{\pd^{n+1}}{\pd t_0^{n+1}}
Z^{1D},
\qquad \forall n\geq 0,
\ee
and the polymer equation:
\be
\label{eq-1d-polymer}
\sum_{n\geq 0}\frac{t_n-\delta_{n,1}}{n!}\lambda^{2n}
\frac{\pd^n}{\pd t_0^n} Z^{1D}=0.
\ee
\end{Theorem}

From the flow equations and the polymer equation,
one can further derive the following Virasoro constraints
for the topological 1D gravity:
\begin{Theorem}
[\cite{ny, zhou6}]
\label{thm-1d-Virasoro-t}
The partition function $Z^{1D}$ satisfies:
\be
\label{eq-1d-Virasoro-t}
L_m (Z^{1D}) =0,\qquad
\forall  m\geq -1,
\ee
where the Virasoro operators $\{L_m\}$ are defined by:
\be
\begin{split}
&L_{-1}=\frac{t_0}{\lambda^2}+
\sum_{n\geq 1}(t_n-\delta_{n,1})\frac{\pd}{\pd t_{n-1}},\\
&L_0=1+\sum_{n\geq 0}(n+1)(t_n-\delta_{n,1})\frac{\pd}{\pd t_{n}},\\
&L_m= \lambda^2\cdot (m+1)! \frac{\pd}{\pd t_{m-1}}
+\sum_{n\geq 1} (t_{n-1}-\delta_{n,2})
\frac{(m+n)!}{(n-1)!}\frac{\pd}{\pd t_{m+n-1}},
\quad m\geq 1.
\end{split}
\ee
Moreover,
the operators $\{L_m\}_{m\geq -1}$ satisfy the Virasoro commutation relations:
\be
[L_m,L_n]=(m-n)L_{m+n},\qquad \forall  m,n\geq -1.
\ee
\end{Theorem}

\subsection{A reformulation of the Virasoro constraints using flow equations}
\label{sec-Virasoro-new}

In this subsection,
we reformulate the Virasoro constraints \eqref{eq-1d-Virasoro-t} in another fashion.
We will need this new version of Virasoro constraints in the next section.

Using the flow equations \eqref{eq-1d-flow},
we easily see that for every $1\leq m_1\leq m-1$ and $m_2:=m-m_1$,
we always have:
\be
\label{eq-cor-flow}
\lambda^2\cdot m!\frac{\pd}{\pd t_{m-1}} Z^{1D}=
\lambda^{2m}\frac{\pd^m}{\pd t_0^m} Z^{1D} =
\lambda^4\cdot m_1!m_2! \frac{\pd^2}{\pd t_{m_1-1}\pd t_{m_2-1}}
Z^{1D},
\ee
thus:
\begin{equation*}
\begin{split}
(m+1)! \frac{\pd Z^{1D}}{\pd t_{m-1}}
=&\bigg(\frac{m+3}{2} + \frac{m-1}{2}\bigg) \cdot m! \frac{\pd Z^{1D}}{\pd t_{m-1}} \\
=&\frac{m+3}{2}\cdot m! \frac{\pd Z^{1D}}{\pd t_{m-1}}
+\half \sum_{\substack{m_1+m_2=m\\m_1,m_2\geq 1}} \lambda^2\cdot m_1!m_2!
\frac{\pd^2 Z^{1D}}{\pd t_{m_1-1}\pd t_{m_2-1}}.
\end{split}
\end{equation*}
Using this relation,
we can reformulate the Theorem \ref{thm-1d-Virasoro-t} in the following way:
\begin{Theorem}
$Z^{1D}$ satisfies the following Virasoro constraints:
\be
L_{m}^{1D}(Z^{1D})=0,
\qquad\forall m\geq -1,
\ee
where the new Virasoro operators $\{L_{m}^{1D}\}_{m\geq -1}$ are defined by:
\be
\label{eq-1d-Virasoro-opr-t}
\begin{split}
L_{-1}^{1D}=&\frac{t_0}{\lambda^2}+
\sum_{n\geq 1}(t_n-\delta_{n,1})\frac{\pd}{\pd t_{n-1}},\\
L_0^{1D}=&1+\sum_{n\geq 0}(n+1)(t_n-\delta_{n,1})\frac{\pd}{\pd t_{n}},\\
L_m^{1D}:=& \frac{m+3}{2}\lambda^2 \cdot m!\cdot \frac{\pd}{\pd t_{m-1}}
+\half \lambda^4\cdot \sum_{\substack{m_1+m_2=m\\m_1,m_2\geq 1}}
m_1!\cdot m_2!\cdot \frac{\pd^2}{\pd t_{m_1-1}\pd t_{m_2-1}} \\
&+\sum_{n\geq 0}\frac{(m+n+1)!}{n!}(t_{n}-\delta_{n,1})\frac{\pd}{\pd t_{m+n}},
\qquad m\geq 1.
\end{split}
\ee
Moreover,
one still has $[L_m^{1D},L_n^{1D}]=(m-n)L_{m+n}^{1D}$ for all $m,n\geq -1$.
\end{Theorem}

\subsection{$Z^{1D}|_{\lambda=1}$ as a tau-function of the KP hierarchy}

It is known in Nishigaki-Yoneya \cite{ny2} that
the partition function $Z^{1D}|_{\lambda=1}$ is a tau-function
of the KP hierarchy.
In this subsection we recall some relevant results.
We will always assume $\lambda=1$ and omit `$|_{\lambda=1}$' for simplicity.

Using the flow equations and Virasoro constraints,
one can derive a cut-and-join representation for the partition function $Z^{1D}$:
\begin{Theorem}
[\cite{wz4}]
We have
$Z^{1D} = \exp(M^{1D}) (1)$,
where:
\be
\begin{split}
M^{1D}=& \half\sum_{k,n\geq 0} \frac{(n+k)!}{n!\cdot k!} t_k t_n \frac{\pd}{\pd t_{n+k-1}}
+ \half \sum_{n\geq 2}\frac{t_{n+1}}{(n+1)!}\sum_{\substack{i+j=n\\i,j\geq 1}} i!\cdot j!
\frac{\pd^2}{\pd t_{i-1}\pd t_{j-1}}\\
&+ \sum_{n\geq 0} \frac{t_{n+2}}{n+2}\cdot \frac{\pd}{\pd t_n} + \half t_0^2 + \half t_1,
\end{split}
\ee
and we use the convention $\frac{\pd}{\pd t_{-1}}:=0$.
\end{Theorem}

Now let us take $t_n = n!\cdot p_{n+1}$ where $p_i$ is the Newton symmetric function of degree $i$.
Then $Z^{1D}\in\Lambda$,
where $\Lambda$ is the bosonic Fock space,
i.e.,
the space of all symmetric functions (see \cite{djm} for an introduction).

Recall that $\Lambda$ has a basis $\{p_\lambda\}_\lambda$,
where for a partition $\lambda=(\lambda_1,\cdots,\lambda_n)$  (with $\lambda_1\geq \cdots\lambda_n >0$)
we denote $p_\lambda := p_{\lambda_1}\cdots p_{\lambda_n}$.
Moreover, we use the convention $p_\emptyset :=1$.
Let $\alpha_n$ ($n\in\bZ$) be the following bosonic operators on $\Lambda$:
\be
\alpha_n :=
\begin{cases}
p_{-n}\cdot, & n<0;\\
0, & n=0;\\
n\frac{\pd}{\pd p_n}, & n>0,
\end{cases}
\ee
then they satisfy the canonical commutation relations:
\be
[\alpha_m, \alpha_n] = m \delta_{m+n,0}.
\ee
After taking $t_n:= n!\cdot p_{n+1}$ for every $n\geq 0$,
the operator $M^{1D}$ becomes:
\be
M^{1D}
= \half \sum_{i+j+k=-2}:\alpha_i\alpha_j\alpha_k: +
\half \sum_{i+j=-2}:\alpha_i\alpha_j: +\half \alpha_{-2}
\ee
where the normal ordering $:\alpha_{i_1}\alpha_{i_2}\cdots \alpha_{i_n}:$
of the bosons $\alpha_{i_1},\cdots, \alpha_{i_n}$ is defined by:
\ben
:\alpha_{i_1}\alpha_{i_2}\cdots \alpha_{i_n}:
= \alpha_{i_{\sigma(1)}}\alpha_{i_{\sigma(2)}}\cdots \alpha_{i_{\sigma(n)}},
\een
where $\sigma\in S_n$ is a permutation such that $\sigma(1)\leq \cdots\leq \sigma(n)$.
Then one finds that $M^{1D}$ is an element in $\widehat{\mathfrak{gl}(\infty)}$
(see eg. \cite[\S 4]{kaz}),
and thus $\exp(M^{1D})\in \widehat{GL(\infty)}$.

In Sato's theory the space of all (formal power series) tau-functions is the $\widehat{GL(\infty)}$-orbit
of the trivial solution $\tau=1$ (see \cite{sa, djm}).
Thus we obtain the result of Nishigaki-Yoneya \cite{ny2} that
$Z^{1D}=\exp(M^{1D})(1)$ is a tau-function of the KP hierarchy,
and the time variables are given by:
\be
x_n := \frac{p_n}{n} = \frac{1}{n!} t_{n-1},
\qquad
n\geq 1.
\ee

\subsection{Quadratic recursion of the $n$-point functions}
\label{eq-qrec-thin}

One way to compute the free energy $F^{1D}$ is to derive recursions
for the (connected) $n$-point functions:
\begin{equation}
\begin{split}
W_{g,n}^{1D}(z_1,\cdots,z_n)
:=&\sum_{j_1,\cdots,j_n\geq 1}
\frac{\pd^n F_g^{1D} (\bm t)}{\pd x_{j_1}\cdots \pd x_{j_n}}\bigg|_{\bm x=0}
\cdot z_1^{-j_1-1} \cdots z_n^{-j_n-1},
\end{split}
\end{equation}
where $(g,n)\not=(0,1)$,
and $x_n = \frac{1}{n!}t_{n-1}$ are the time variables of the KP hierarchy.
Then the following quadratic recursion formula for these $n$-point functions
emerges naturally from the Virasoro constraints \eqref{eq-1d-Virasoro-t}:
\begin{Proposition}
[\cite{zhou5}]
Denote $W_{0,1}^{1D}(z_1):=z_1^{-1}$,
then all other $W_{g,n}^{1D}$ can be recursively computed by:
\be
\label{eq-toporec-1}
\begin{split}
& W_{g, n+1}^{1D}\left(z_{0}, z_{1}, \cdots, z_{n}\right) \\
=& \sum_{j=1}^{n} D_{z_{0}, z_{j}} W_{g, n}^{1D}\left(z_{1},
\cdots, z_{n}\right)+E_{z_{0}, u, v} W_{g-2, n+2}^{1D}\left(u, v, z_{1}, \cdots, z_{n}\right) \\
&+\sum_{g_{1}+g_{2}=g-1 \atop I \sqcup J=[n]} E_{z_{0}, u, v}
\left(W_{g_{1},|I|+1}^{1D}\left(u, z_{I}\right) \cdot
W_{g_{2},|J|+1}^{1D}\left(v, z_{J}\right)\right).
\end{split}
\ee
where $[n]:=\{1,2,\cdots,n\}$,
$z_{I}:=\left(z_{i_{1}}, \cdots, z_{i_{k}}\right)$ for a set of indices
$I=\left\{i_{1}, \cdots, i_{k}\right\}$,
and the operators $D_{z_{0}, z_{j}}$ and $E_{z_{0}, u, v}$ are defined by:
\begin{equation*}
\begin{split}
&D_{z_{0}, z_{j}} f\left(z_{j}\right):=
\frac{f\left(z_{0}\right)-f\left(z_{j}\right)}{z_{0}\left(z_{0}-z_{j}\right)^{2}}
-\frac{1}{z_{0}\left(z_{0}-z_{j}\right)}
\frac{\partial}{\partial z_{j}} f\left(z_{j}\right),\\
&E_{z_{0}, u, v} f(u, v):= \frac{1}{z_{0}} \cdot
\big( \lim _{u \rightarrow v} f(u, v) \big)\big|_{v=z_{0}}.
\end{split}
\end{equation*}
\end{Proposition}

\subsection{Fat genus expansion, fat spectral curve, and emergence of E-O topological recursion}
\label{sec-toporec-fat}

The quadratic recursion \eqref{eq-toporec-1} looks similar to but is not an Eynard-Orantin topological recursion
\cite{ce, eo}.
In order to obtain an E-O topological recursion,
one needs to consider the fat genus expansion and fat Virasoro constraints for $Z^{1D}$.
This has been done in \cite{zhou5},
and in this subsection we recall the main results.

The topological 1D gravity is the special case $N=1$ of the Hermitian one-matrix models,
thus it admits a genus expansion by fat graphs,
see \cite{zhou9}.
Roughly speaking,
a fat graph is (the equivalence class of) the $1$-skeleton of a cell-decomposition of a Riemann surface,
and the genus of a fat graph is the genus of the surface.
Then one has a fat genus expansion and a family of fat Virasoro constraints for $Z^{1D}$
which are different from those we have seen in \S \ref{sec-pre-Virasoro}.
Here we only recall some necessary results,
for details of the fat and thin emergent geometry, see \cite{zhou5, zhou8, zhou9};
and see the book \cite{lz} for an introduction for fat graphs and Hermitian matrix models.

Take $\lambda=1$, and let
\be
F^{1D} (\bm t) = \sum_{g\geq 0} F_g^{\text{fat}} (\bm t)
\ee
be the fat genus expansion of the free energy (see \cite{zhou9}).
Let $x_n = \frac{1}{n!}t_{n-1}$,
and define the fat $n$-point function of genus $g$ to be:
\be
W_{g,n}^{\text{fat}}(z_1,\cdots,z_n)
:=\sum_{j_1,\cdots,j_n\geq 1}
\frac{\pd^n F_g^{\text{fat}} (\bm t)}{\pd x_{j_1}\cdots \pd x_{j_n}}\bigg|_{\bm x=0}
\cdot z_1^{-j_1-1} \cdots z_n^{-j_n-1}
\ee
for $g\geq 0$ and $n\geq 1$.
Notice that for the fat genus expansion,
the coefficients in the case $(g,n)=(0,1)$ are given by the Catalan numbers
(see eg. \cite{dmss, mu}),
while in the thin genus expansion \eqref{eq-genusexp-F} there is no terms of type $(g,n)=(0,1)$.

\begin{Example}
The following examples are known in literatures
(see eg. \cite{zhou5}):
\be
\begin{split}
&W_{0,1}^{\text{fat}} (z_1) = \half (z_1-\sqrt{z_1^2-4}), \\
&W_{0,2}^{\text{fat}} (z_1,z_2) = \frac{1}{2(z_1-z_2)^2}
\bigg(
-1 + \frac{z_1z_1-4}{\sqrt{(z_1^2-4)(z_2^2-4)}}
\bigg).
\end{split}
\ee
They can be computed using the fat Virasoro constraints.
\end{Example}

In order to derive a topological recursion in the sense of Eynard-Orantin,
one needs a spectral curve first.
In \cite{zhou8},
the second author has shown that the following
fat special deformation emerges from the fat Virasoro constraints
(\cite[(78)]{zhou8}):
\be
y= \frac{1}{\sqrt{2}}
\sum_{n\geq 1} \frac{t_{n-1}-\delta_{n,2}}{(n-1)!}\cdot z^{n-1}
+\frac{\sqrt{2}}{z}
+\sqrt{2}\cdot\sum_{n\geq 1}
n!  \frac{\pd F_0^{\text{fat}}}{\pd t_{n-1}} \cdot z^{-n-1},
\ee
which is a family of curves on the $(y,z)$-plane parametrized by $(t_0, t_1,\cdots)$.
By setting the coupling constants to zero
(and perform a suitable rescaling),
one obtains the following fat spectral curve:
\be
\label{eq-fat-curve}
z^2-4y^2 =4.
\ee
This is the special case $N=1$ of the fat spectral curve for the Hermitian one-matrix models
(at finite $N$),
see \cite[\S 3]{zhou8} for details and the quantum deformation theory for this spectral curve.
In physics literatures,
this curve is known as the semi-circle law for the Hermitian matrix models,
see eg. \cite{mar} for a treatment from the large $N$ point of view.

Furthermore,
one can show that the Eynard-Orantin topological recursion on this spectral curve emerges naturally
from the fat Virasoro constraints:
\begin{Theorem}
[\cite{zhou5}]
Let $p,p_i$ be points on this spectral curve \eqref{eq-fat-curve}, and denote:
\be
\begin{split}
&\omega_{0,1}^{\text{fat}} (p_1) = y(p_1) dz(p_1),\\
&\omega_{0,2}^{\text{fat}} (p_1,p_2) =
\frac{1}{2(z_{1}(p)-z_{2}(p))^{2}}\bigg(1+\frac{z_{1}(p) z_{2}(p)-4 }{4 y_{1}(p) y_{2}(p)}\bigg) d z_{1}(p) d z_{2}(p)
\end{split}
\ee
and
\be
\omega_{g,n}^{\text{fat}} (z_1,\cdots,z_n) := W_{g,n}^{\text{fat}}(z_1,\cdots,z_n)dz_1\cdots dz_n
\ee
for $2g-2+n>0$.
Then the following E-O topological recursion emerges from the fat Virasoro constraints for $F^{1D}$:
\be
\label{eq-EO-fat}
\begin{split}
\omega_{g, n+1}^{\text{fat}}(p_{0}, \cdots, p_{n})=
&(\Res_{p \rightarrow p_{+}}+\Res_{p \rightarrow p_{-}})
\bigg[K( p_0 , p ) \bigg(\omega_{g-1, n+2}^{\text{fat}}(p, \sigma(p), p_{[n]})\\
&+\sum_{g_{1}+g_{2}=g \atop I \sqcup J=[n]}^{s}
\omega_{g_{1},|I|+1}^{\text{fat}}(p, p_{I}) \omega_{g_{2},|J|+1}^{\text{fat}}(\sigma(p), p_{J})\bigg)\bigg],
\end{split}
\ee
for $2g-2+n>0$,
where $p_\pm$ is the two branch point of the spectral curve,
$\sigma$ is the involution $\sigma(z,y)=\sigma(z,-y)$ near the branch points,
and we denote $p_I:=(p_{i_1,\cdots,p_{i_k}})$ for $I=\{i_1,\cdots,i_k\}$.
The `s' on $\sum$ means excluding terms involving $\omega_{0,1}^{\text{fat}}$ in this summation.
The recursion kernel $K$ is given by:
\be
K\left(p_{0}, p\right)=\frac{\int_{\sigma(p)}^{p} B^{H}
\left(p_{0}, p\right)}{2(y(p)-y(\sigma(p))) d z(p)}=\frac{d z_{0}}{4 y_{0}\left(z_{0}-z\right) d z}.
\ee
\end{Theorem}

\section{Affine Coordinates for $Z^{1D}|_{\lambda=1}$ and Some Applications}
\label{sec-mainresult}

In this section,
we study the tau-function $Z^{1D}|_{\lambda=1}$ of the KP hierarchy.
In particular,
we discuss some applications of its affine coordinates on the Sato Grassmannian.
We describe this tau-function as a Bogoliubov transform in the fermionic picture,
and as a summation of Schur functions in the bosonic picture,
and then show how to compute the connected $n$-point functions using the affine coordinates.
Moreover,
we find a family of Kac-Schwarz operators for this tau-function.
Some of these results will not be used in the computation of $\chi(\Mbar_{g,n})$ in the next section,
but we describe them for completeness.

\subsection{Affine coordinates for the tau-function $Z^{1D}|_{\lambda=1}$}
\label{sec-1d-affinecoord}

First let us recall the affine coordinates of the tau-function $Z^{1D}|_{\lambda=1}$ of the KP hierarchy.
From now on,
we will always assume that $\lambda=1$,
and omit `$|_{\lambda=1}$' for simplicity.
This evaluation dose not lose any information since the $\lambda$-dependence can be read off
from $Z^{1D}|_{\lambda=1}$ by the selection rule \eqref{eq-selection-1d}.

In Kyoto School's approach to the KP hierarchy,
there are three equivalent descriptions to a tau-function:
\begin{itemize}
\item[1)]
An element on the Sato Grassmannian;
\item[2)]
A vector in the fermionic Fock space;
\item[3)]
A vector in the bosonic Fock space.
\end{itemize}

Here let us briefly recall the construction of the big cell $Gr_{(0)}$ of the Sato Grassmannian.
Let $H$ be the infinite-dimensional vector space:
\ben
H:=\big\{\text{formal series }\sum_{n\in \bZ} a_n z^{n}
\big| a_n=0 \text{ for } n>>0
\big\},
\een
Then one has $H=H_+ \oplus H_-$,
where
$H_+ = \bC [z]$ and
$H_- = z^{-1} \bC [[z^{-1}]]$.
Denote by $\pi_\pm :H\to H_\pm$ the two natural projections.
The big cell $Gr_{(0)}$ of the Sato Grassmannian consists of all linear subspaces $U\subset H$
such that the projection $\pi_+: U \to H_+$ is an isomorphism of vector spaces.

An element $U\subset H$ of the big cell $Gr_{(0)}$ of the Sato Grassmannian is
spanned by a basis of the form:
\be
\big\{
\tilde f_n
=z^{n} + \sum_{j<n}\tilde a_{n,j} z^{j}
\big\}_{n\geq 0},
\ee
which is called an admissible basis.
Among various admissible basis for $U$,
there is a unique one of the following form
(see eg. \cite{by} or \cite[\S 3]{zhou1}):
\be
\big\{
f_n
=z^{n} + \sum_{m\geq 0} a_{n,m} z^{-m-1}
\big\}_{n\geq 0},
\ee
called the normalized basis for $U$.
The coefficients $\{a_{n,m}\}_{n,m\geq 0}$ are called the affine coordinates for $U$.
Given an element $U\in Gr_{(0)}$,
one can construct a vector $|U\rangle$ in the fermionic Fock space
as a semi-infinite wedge product,
and then construct a vector $\tau_U$ in the bosonic Fock space using the boson-fermion correspondence.
The element $\tau_U$ constructed in this way is a tau-function of the KP hierarchy.
Conversely,
given any tau-function of the KP hierarchy,
one can associate an element in $Gr_{(0)}$ to it.
See \cite{djm} for details of boson-fermion correspondence.
We will give a brief review in Appendix \ref{sec-app-pre}.

In particular,
the tau-function $Z^{1D}$ can be regarded as an element in the big cell $Gr_{(0)}$
of the Sato Grassmannian.
Let us denote by $U^{1D}\subset H$ the point on $Gr_{(0)}$ corresponding to $Z^{1D}$,
and $\{a_{n,m}^{1D}\}_{n,m\geq 0}$ the affine coordinates of $U^{1D}$.
Then:
\begin{Theorem}
[\cite{zhou7}]
\label{thm-affinecoord-1D}
The affine coordinates $\{a_{n,m}^{1D}\}_{n,m\geq 0}$ are given by:
\be
\label{eq-affinecoord-1D}
a_{n,m}^{1D} = \begin{cases}
m!!, & \text{if $n=0$ and $m$ is odd;}\\
0, & \text{otherwise.}
\end{cases}
\ee
Or equivalently,
$U^{1D}\subset H$ is spanned by the normalized basis $\{f_n^{1D}(z)\}_{n\geq 0}$
where:
\be
\label{eq-nbasis-1D}
f_n^{1D}(z)=
\begin{cases}
1 +\sum\limits_{k\geq 0} (2k+1)!! \cdot z^{-2k-2}, & n=0;\\
z^n, & n>0.
\end{cases}
\ee
\end{Theorem}

Theorem \ref{thm-affinecoord-1D} can be proved by various methods.
In \cite{zhou7},
the second author have computed the affine coordinates for the partition function $Z_N$ of
the Hermitian one-matrix models (at finite $N$)
using a method inspired by the representation theory \cite{iz2},
and the above theorem can be obtained by simply taking $N=1$ in \cite[(64)]{zhou7}.
In another paper \cite{zhou2},
the affine coordinates for the partition function of counting Grothendieck's dessin d'enfants
which includes the Hermitian one-matrix models as a special case,
were computed using the fermionic reformulation of the Virasoro constraints.

In Appendix \ref{sec-proof-comb},
we will give another proof of this result using a combinatorial method.
This proof is rather simple,
but reveals some interesting relations between symmetric functions and Feynman graphs,
and this fits into our purpose of relating quantum field theories to integrable hierarchies.

\subsection{A fermionic representation and a bosonic representation for $Z^{1D}$}
\label{sec-fermi-1D}

By applying \cite[Theorem 3.1]{zhou1} to the affine coordinates \eqref{eq-affinecoord-1D},
one obtains the following fermionic representation of
the tau-function $Z^{1D}$
(see \cite{zhou1} for notations,
or Appendix \ref{sec-app-pre} for a brief introduction):
\begin{Proposition}
In the fermionic Fock space,
the tau-function $Z^{1D}$ is given by the following Bogoliubov transform
of the fermionic vacuum $|0\rangle$:
\be
Z^{1D} = e^{A^{1D}} |0\rangle,
\ee
where $A^{1D}$ is the following quadratic operator of the fermionic creators $\{\psi_r,\psi_r^*\}_{r>0}$:
\be
A^{1D}:= \sum_{m,n\geq 0} a_{n,m}^{1D} \psi_{-m-\half} \psi_{-n-\half}^*=
\sum_{k=0}^\infty (2k+1)!! \cdot \psi_{-2k-\frac{3}{2}} \psi_{-\half}^*.
\ee
Expanding the exponential,
then we have:
\be
Z^{1D}=
|0\rangle +\sum_{k=0}^\infty (2k+1)!! \cdot \psi_{-2k-\frac{3}{2}} \psi_{-\half}^* |0\rangle.
\ee
\end{Proposition}

Then applying the boson-fermion correspondence,
one obtains a bosonic representation for $Z^{1D}$ as a summation of symmetric functions
(after taking $t_{n} = n!\cdot p_{n+1}$ for every $n\geq 0$,
where $p_{n+1}$ is the Newton symmetric function of degree $n+1$):
\be
\label{eq-Z-symm}
\begin{split}
Z^{1D} (\bm t)=&
1+\sum_{k=0} ^\infty (2k+1)!!\cdot s_{(2k+1|0)}\\
=& 1+\sum_{k=0} ^\infty (2k+1)!!\cdot h_{2k+2},
\end{split}
\ee
where $s_{(2k+1|0)}$ is the Schur function associated to the partition $(2k+1|0)=(1^{2k+2})$,
and $h_{2k+2}$ is the complete symmetric function of degree $2k+2$.
See \cite{mac} for an introduction to the symmetric functions.

\subsection{Representing $F^{1D}$ in terms of Schur functions}

No let us represent the free energy $F^{1D}=\log(Z^{1D})$ in terms of Schur functions.
From \eqref{eq-Z-symm} one sees that
\ben
F^{1D}(\bm t)= \sum_{n\geq 1}\frac{(-1)^n}{n}
\bigg(
\sum_{k=0}^\infty (2k+1)!! \cdot h_{2k+2}
\bigg)^n
\een
by expanding the logarithm. Thus:
\begin{Proposition}
We have:
\begin{equation*}
\begin{split}
F^{1D}(\bm t)=&
\sum_{n\geq 1}\sum_{k_1,\cdots, k_n\geq 0}
\frac{(-1)^n}{n}(2k_1+1)!!\cdots (2k_n+1)!!
\cdot h_{2k_1+2}\cdots h_{2k_n+2}\\
=&\sum_{n\geq 1}\sum_{k_1,\cdots, k_n\geq 0}
\frac{(-1)^n}{n}(2k_1+1)!!\cdots (2k_n+1)!!
\cdot s_{(2k_1+1|0)}\cdots s_{(2k_n+1|0)},
\end{split}
\end{equation*}
where $t_n=n!\cdot p_{n+1}$ for every $n\geq 0$.
\end{Proposition}

It is well-known that the product of two Schur functions is given by
the following Littlewood-Richardson rule:
\ben
s_\lambda s_\mu =\sum_{\nu} c_{\lambda\mu}^\nu s_\nu,
\een
where the summation is over partitions $\nu$ with $|\nu|=|\lambda|+|\mu|$,
and the coefficients $c_{\lambda\mu}^\nu$ can be obtained by counting certain operations
on the corresponding Young diagrams
(see eg. \cite[\S 1]{mac}).
Thus we can rewrite the above formula as:
\begin{Corollary}
We have:
\begin{equation*}
\begin{split}
F^{1D}(\bm t)=&
\sum_{n\geq 1}\sum_{k_1,\cdots, k_n\geq 0} \sum_{\mu_1,\cdots,\mu_{n-1}}
\frac{(-1)^n}{n}(2k_1+1)!!\cdots (2k_n+1)!!\\
&\qquad\cdot
c_{(2k_1+1|0),(2k_2+1|0)}^{\mu_1}
c_{\mu_1,(2k_3+1|0)}^{\mu_2}\cdots
c_{\mu_{n-2},(2k_n-1|0)}^{\mu_{n-1}} s_{\mu_{n-1}},
\end{split}
\end{equation*}
where $t_n=n!\cdot p_{n+1}$ ($n\geq 0$),
and $c_{\lambda\mu}^\nu$ are the Littlewood-Richardson coefficients.
\end{Corollary}

\subsection{A Kac-Schwarz operator for $Z^{1D}$}
\label{sec-KS}

In this subsection we find a Kac-Schwarz operator
for the tau-function $Z^{1D}$.

First let us briefly recall the definition of a Kac-Schwarz operator \cite{ks, sc}.
Let $U\subset H=\bC[z]\oplus z^{-1}\bC[[z^{-1}]]$ be an element of the big cell $Gr_{(0)}$ of the Sato Grassmannian.
A Kac-Schwarz operator for $U$ is a differential operator $Q$ in $z$,
satisfying the following condition:
\ben
Q( U)\subset  U.
\een
Among all the Kac-Schwarz operators,
the most interesting ones are those which takes a Laurent polynomial of degree $n$
to a Laurent polynomial of degree $n+1$;
i.e.,
those such that there exists a basis for $U$ of the form:
\ben
\{u_n:=c_n z^n + \text{low order terms}\}_{n=0,1,2,\cdots}
\een
where $c_n\not=0$,
satisfying $Q(u_n)=u_{n+1}$ for every $n\geq 0$.
In this case,
the vector space $U$ is spanned by $\{Q^k(u_0)\}_{k\geq 0}$.

Now let us consider the case of the topological 1D gravity.
Let $U^{1D}$ be the element on the big cell of the Sato Grassmannian corresponding to the tau-function $Z^{1D}$,
and we already know that $U^{1D}$ is spanned by $\{f_n^{1D}\}_{n\geq 0}$ given by \eqref{eq-nbasis-1D}.
Now let us define:
\be
Q^{1D} := \pd_z +z-z^{-1},
\ee
then we have:
\begin{Theorem}
The operator $Q^{1D}= \pd_z +z-z^{-1}$ is a Kac-Schwarz operators for $U^{1D}$.
Moreover,
we have:
\be
U^{1D} = \text{span} \{(Q^{1D})^n( f_0^{1D})\}_{n\geq 0}.
\ee
\end{Theorem}
\begin{proof}
By direct computations we easily check that:
\begin{equation*}
\begin{split}
Q^{1D} ( f_0^{1D})=&
-\sum_{k=0}^\infty (2k+2)\cdot(2k+1)!!z^{-2k-3}+z
+\sum_{k=0}^\infty (2k+1)!!z^{-2k-1}\\
&-z^{-1}-\sum_{k=0}^\infty (2k+1)!!z^{-2k-3}\\
=&z,
\end{split}
\end{equation*}
i.e., $Q^{1D}( f_0^{1D})= f_1^{1D}$.
Moreover,
for $n\geq 1$ we have:
\begin{equation*}
Q^{1D} ( f_n^{1D})= nz^{n-1}+z^{n+1}-z^{n-1}=
 f_{n+1}^{1D}+(n-1)  f_{n-1}^{1D},
\end{equation*}
thus by induction $(Q^{1D})^n( f_0^{1D})$ is a polynomial in $z$ of degree $n$ for every $n\geq 1$.
Then the conclusion holds.
\end{proof}

\begin{Remark}
Now let us take the classical limit of this Kac-Schwarz operator
(and then perform an additional rescaling) by:
\ben
z\mapsto \frac{z}{\sqrt{2}}, \qquad\qquad \pd_z \mapsto y,
\een
and in this way we obtain a plane curve:
\be
y= -\frac{z}{\sqrt{2}}+\frac{\sqrt{2}}{z},
\ee
which is called the signed Catalan curve.
This is the thin spectral curve of the topological 1D gravity
emerging from the thin Virasoro constraints (see \cite[(311)]{zhou6}).
See \cite[\S 10-\S 11]{zhou6} for the quantum deformation theory for this curve.

\end{Remark}

\subsection{From Kac-Schwarz operators to Virasoro operators}
\label{sec-KS-Virasoro}

In this subsection,
we explain how to recover the Virasoro constraints of the topological 1D gravity
from the Kac-Schwarz operator $Q^{1D}$ defined above.

For every $n\geq -1$,
define:
\be
Q_n^{1D}:= - z^{n+1} Q^{1D}
= -z^{n+1}\pd_z -z^{n+2} +z^n.
\ee
Then one can check that they satisfy the Virasoro commutation relations:
\be
\begin{split}
[Q_m^{1D},Q_n^{1D}] =&
z^m(\pd_z +z+z^{-1})\circ z^n Q^{1D}
-z^n(\pd_z +z+z^{-1})\circ z^m Q^{1D}\\
=& (n-m) z^{m+n-1} Q^{1D}\\
=&(m-n) Q_{m+n}^{1D},
\qquad \forall m,n\geq -1.
\end{split}
\ee
Moreover,
it is clear that:
\be
Q_n^{1D} ( U^{1D})\subset  U^{1D},
\qquad \forall n\geq -1,
\ee
i.e., $\{Q_n^{1D}\}_{n\geq -1}$ are all Kac-Schwarz operators for $U^{1D}\in Gr_{(0)}$.

The operators $\{Q_n^{1D}\}_{n\geq -1}$ can be translated into operators
on the bosonic Fock space via the boson-fermion correspondence.
Now let us explain this.
A differential operator $z^l \pd_z^m$ is an element in $\mathfrak{gl}(\infty)$ (see \cite{kaz}),
and it defines an action $\widehat{z^l \pd_z^m}$ on semi-infinite wedge products
(i.e., vectors in the fermionic Fock space) by:
\ben
&&\widehat{z^l \pd_z^m}(z^{k_1}\wedge z^{k_2}\wedge \cdots)\\
&=&
\sum_{n=0}^\infty m!(\binom{k_n}{m}-\delta_{l,m}\binom{n}{m})
z^{k_1}\wedge \cdots\wedge
z^{k_{n-1}}\wedge z^{k_n+l-m}\wedge z^{k_{n+1}}\wedge \cdots,
\een
where $k_i\in \bZ$ for every $i$ and $k_{i+1}=k_i+1$ for $i>>0$.
This defines a central extension $\widehat{\mathfrak{gl}(\infty)}$ of $\mathfrak{gl}(\infty)$.
Then by the boson-fermion correspondence,
they become operators on the bosonic Fock space.
According to \cite[Prop. 6.5]{kaz},
one has
(after some change of notations):
\be
\widehat{z^m}=\alpha_m,
\qquad\quad
\widehat{z^n\pd_z}=
-\half\sum_{i\in \bZ}:\alpha_i \alpha_{n-1-i}:
-\frac{n-2}{2}\alpha_{n-1},
\ee
where $:\alpha_i \alpha_{n-1-i}:$ means the normal-ordering of bosons,
and $\alpha_m$ are the bosonic operators given by:
\be
\alpha_m=\begin{cases}
m\frac{\pd}{\pd p_m},
& m>0;\\
1, & m=0;\\
p_{-m}, & m<0.
\end{cases}
\ee
Here $p_m= mx_m$ (and recall that $x_m= \frac{1}{m!}t_{m-1}$ where $\{t_n\}_{n\geq 0}$ are the coupling constants
and $\{x_m\}_{m\geq 1}$ are time variables).
Therefore on the bosonic Fock space,
the operators $Q_n^{1D} = -z^{n+1}\pd_z -z^{n+2} +z^n$ ($n\geq -1$) becomes:
\begin{equation*}
\begin{split}
Q_n^{1D} =& -\alpha_{n+2} + \frac{n+1}{2} \alpha_n
+\half \sum_{i\in \bZ}:\alpha_i \alpha_{n-i}:\\
=&-\alpha_{n+2} + \frac{n+3-\delta_{n,0}}{2} \alpha_n
+\half \sum_{i\not= 0,n}:\alpha_i \alpha_{n-i}:.
\end{split}
\end{equation*}
Now let us rewrite $Q_n^{1D}$ in terms of the coupling constants $\{t_n=n!\cdot p_{n+1}\}_{n\geq 0}$,
and then it follows that:
\begin{Theorem}
Under the boson-fermion correspondence,
we have:
\be
\widehat{Q_n^{1D}} = L_n^{1D},
\qquad
\forall n\geq -1,
\ee
where $Q_n^{1D}$ are the Kac-Schwarz operators:
\ben
Q_n^{1D} := - z^{n+1} Q^{1D} = -z^{n+1}\pd_z -z^{n+2} +z^n,
\een
and $\{L_n^{1D}\}_{n\geq -1}$ are the Virasoro operators defined by \eqref{eq-1d-Virasoro-opr-t}.
\end{Theorem}

Notice that $Q_n^{1D}(U^{1D}) \subsetneqq U^{1D}$,
and then $Q_n^{1D}(U^{1D})\not\in Gr_{(0)}$.
This condition implies the Virasoro constraints
$L_n^{1D}(Z^{1D})=0$ for $n\geq -1$.

\subsection{Computing the $n$-point functions using Zhou's formula}
\label{sec-npt-1d}

In \cite[\S 5]{zhou1},
the second author has derived a formula to compute the connected $n$-point functions
associated to an arbitrary tau-function of the KP hierarchy in terms of its affine coordinates.
In this subsection we apply it to the topological 1D gravity.

Let $\tau=\tau(\bm x)$ be a tau-function of the KP hierarchy with respect to
the time variables $\bm x=(x_1,x_2,\cdots)$,
and define the (all-genera) connected $n$-point function associated to $\tau$ to be:
\be
G^{(n)}(z_1,\cdots,z_n)
:=\sum_{j_1,\cdots,j_n\geq 1}
\frac{\pd^n \log \tau (\bm x)}{\pd x_{j_1}\cdots \pd x_{j_n}}\bigg|_{\bm x=0}
\cdot z_1^{-j_1-1} \cdots z_n^{-j_n-1}.
\ee
Then one has:
\begin{Theorem}
[\cite{zhou1}]
Let $\{a_{n,m}\}_{n,m\geq 0}$ be the affine coordinates for $\tau$ on the big cell of the Sato Grassmannian,
and denote:
\be
A(\xi,\eta):= \sum_{m,n\geq 0}
a_{n,m} \xi^{-n-1} \eta^{-m-1}.
\ee
Then:
\be
\label{eq-thm-npt}
G^{(n)}(z_1,\cdots,z_n)
=(-1)^{n-1} \sum_{n\text{-cycles }\sigma}
\prod_{i=1}^{n} \widehat A (z_{\sigma(i)}, z_{\sigma(i+1)})
-\frac{\delta_{n,2}}{(z_1-z_2)^2},
\ee
where $\sigma(n+1):=\sigma(1)$, and:
\be
\widehat A (z_i,z_j)=\begin{cases}
i_{z_i,z_j}\frac{1}{z_i-z_j} + A(z_i,z_j), & i<j;\\
A(z_i,z_i), & i=j;\\
i_{z_j,z_i}\frac{1}{z_i-z_j} + A(z_i,z_j), & i>j,
\end{cases}
\ee
and $i_{\xi,\eta}\frac{1}{\xi+\eta}:= \sum_{k\geq 0} (-1)^k \xi^{-1-k}\eta^k$.
\end{Theorem}

The notation $i_{z_i,z_j}\frac{1}{z_i-z_j}$ indicates how to expand $\frac{1}{z_i-z_j}$
as a formal series in $z_i$ and $z_j$.
In what follows we will denote it by $\frac{1}{z_i-z_j}$ for simplicity.
Now one can use the formula \eqref{eq-thm-npt} to compute
the connected $n$-point functions:
\begin{equation}
\begin{split}
G_{(n)}^{1D}(z_1,\cdots,z_n)
:=&\sum_{j_1,\cdots,j_n\geq 1}
\frac{\pd^n F^{1D} (\bm t)}{\pd x_{j_1}\cdots \pd x_{j_n}}\bigg|_{\bm x=0}
\cdot z_1^{-j_1-1} \cdots z_n^{-j_n-1}\\
=&\sum_{j_1,\cdots,j_n\geq 1}
\frac{\pd^n F^{1D} (\bm t)}{\pd t_{j_1-1}\cdots \pd t_{j_n-1}}\bigg|_{\bm t=0}
\cdot \frac{j_1!\cdots j_n!}{z_1^{j_1+1} \cdots z_n^{j_n+1}}
\end{split}
\end{equation}
of the topological 1D gravity.
Denote:
\be
\label{eq-Agenerating-1d}
A^{1D}(\xi,\eta)=\sum_{n,m\geq 0} a_{n,m}^{1D} \xi^{-n-1} \eta^{-m-1}
=\xi^{-1}\cdot\sum_{k=0}^\infty (2k+1)!!\cdot \eta^{-2k-2}.
\ee
Then by \eqref{eq-thm-npt} we obtain the following formula for $G_{(n)}^{1D}$:
\begin{Proposition}
\label{prop-thm-npt-1d}
Let $A^{1D}(\xi,\eta)$ be given by \eqref{eq-Agenerating-1d},
then we have:
\be
\label{eq-thm-npt-1d}
G_{(n)}^{1D}(z_1,\cdots,z_n)
=(-1)^{n-1} \sum_{n\text{-cycles }\sigma}
\prod_{i=1}^{n}
\widehat A^{1D} (z_{\sigma(i)}, z_{\sigma(i+1)})
-\frac{\delta_{n,2}}{(z_1-z_2)^2},
\ee
where $\widehat A^{1D} (z_{\sigma(i)}, z_{\sigma(i+1)}):=
A^{1D} (z_{\sigma(i)}, z_{\sigma(i+1)}) + \frac{1}{z_{\sigma(i)}- z_{\sigma(i+1)}}$.
\end{Proposition}

\begin{Example}
For $n=1$,
the above formula gives:
\begin{equation*}
G_{(1)}^{1D} (z) = A^{1D} (z,z)
=\sum_{k=0}^\infty (2k+1)!! \cdot z^{-2k-3},
\end{equation*}
and the first a few terms are:
\begin{equation*}
\begin{split}
G_{(1)}^{1D}(z) =&
\frac{1}{z^3}+
+\frac{3}{z^5} + \frac{15}{z^{7}} + \frac{105}{z^{9}} + \frac{945}{z^{11}}
+ \frac{10395}{z^{13}} + \frac{135135}{z^{15}} +
\frac{2027025}{z^{17}}\\
&
 + \frac{34459425}{z^{19}}
+ \frac{654729075}{z^{21}} + \frac{13749310575}{z^{23}} +
\frac{316234143225}{z^{25}}\\
&
 + \frac{7905853580625}{z^{27}}
 + \frac{213458046676875}{z^{29}} +
\frac{6190283353629375}{z^{31}}\\
&
 + \frac{191898783962510625}{z^{33}} +
\frac{6332659870762850625}{z^{35}}
 + \cdots.
\end{split}
\end{equation*}
And for $n=2$,
there is only one $2$-cycle $(1,2)$,
thus the above formula gives:
\begin{equation*}
\begin{split}
G_{(2)}^{1D}(z_1,z_2)
=& - \big( A^{1D}(z_1,z_2)+\frac{1}{z_1-z_2}\big)
\big( A^{1D}(z_2,z_1)+\frac{1}{z_2-z_1}\big)
-\frac{1}{(z_1-z_2)^2}\\
=&\frac{A^{1D}(z_1,z_2)-A^{1D}(z_2,z_1)}{z_1-z_2}
-A^{1D}(z_1,z_2) A^{1D}(z_2,z_1)\\
=&\sum_{k,l\geq 0}\bigg( (2k+2l+3)!! - (2k+1)!! \cdot (2l+1)!!
\bigg) z_1^{-2k-3} z_2^{-2l-3}\\
&+\sum_{k,l\geq 0} (2k+2l+1)!!
\cdot z_1^{-2k-2} z_2^{-2l-2}.
\end{split}
\end{equation*}
The first a few terms are:
\begin{equation*}
\begin{split}
G_{(2)}^{1D}(z_1,z_2)=&
\frac{1}{z_1^2 z_2^2}
+ \frac{3}{z_1^2 z_2^4}
+\frac{2}{z_1^3 z_2^3}
+\frac{3}{z_1^4 z_2^2}
+\frac{15}{z_1^2 z_2^6}
+\frac{12}{z_1^{3} z_2^{5}}
+\frac{15}{z_1^{4} z_2^{4}}
+\frac{12}{z_1^{5} z_2^{3}}\\
&+\frac{15}{z_1^6 z_2^2}
 +\frac{105}{z_1^{2} z_2^{8}}
+\frac{90}{z_1^{3} z_2^{7}}
+\frac{105}{z_1^{4} z_2^{6}}
+\frac{96}{z_1^{5} z_2^{5}}
+\frac{105}{z_1^{6} z_2^{4}}
+\frac{90}{z_1^{7} z_2^{3}}
+\frac{105}{z_1^{8} z_2^{2}}\\
&+\frac{945}{z_1^{2} z_2^{10}}
+\frac{840}{z_1^{3} z_2^{9}}
 +\frac{945}{z_1^{4} z_2^{8}}
+\frac{900}{z_1^{5} z_2^{7}}
+\frac{945}{z_1^{6} z_2^{6}}
+\frac{900}{z_1^{7} z_2^{5}}
+\frac{945}{z_1^{8} z_2^{4}}
+\frac{840}{z_1^{9} z_2^{3}}\\
&+\frac{945}{z_1^{10} z_2^{2}}
+\frac{10395}{z_1^{2} z_2^{12}}
 +\frac{9450}{z_1^{3} z_2^{11}}
+\frac{10395}{z_1^{4} z_2^{10}}
+\frac{10080}{z_1^{5} z_2^{9}}
+\frac{10395}{z_1^{6} z_2^{8}}
+\frac{10170}{z_1^{7} z_2^{7}}\\
&+\frac{10395}{z_1^{8} z_2^{6}}
+\frac{10080}{z_1^{9} z_2^{5}}
+\frac{10395}{z_1^{10} z_2^{4}}
+\frac{9450}{z_1^{11} z_2^{3}}
+\frac{10395}{z_1^{12} z_2^{2}}
+\cdots.
\end{split}
\end{equation*}
From these data one sees that (after taking $\lambda=1$):
\begin{equation*}
\begin{split}
F^{1D}=& (x_2+3x_4 +15x_6 +105x_8 + 945 x_{10}+ 10395 x_{12}+\cdots)\\
&+(\half x_1^2 + 3x_1x_3 + x_2^2 + 15 x_1x_5 + 12x_2x_4 +\frac{15}{2} x_3^2+\cdots)
+\cdots\\
=& (\half t_1+\frac{1}{8}t_3 +\frac{1}{48}t_5 +\frac{1}{384}t_7
+ \frac{1}{3840}t_9+ \frac{1}{46080}t_{11}+\cdots)\\
&+(\half t_0^2 + \half t_0t_2 + \frac{1}{4}t_1^2 + \frac{1}{8} t_0t_4
+\frac{1}{4}t_1t_3 + \frac{5}{24}t_2^2 +\cdots)+\cdots.
\end{split}
\end{equation*}

The $n$-point functions $G_{(n)}^{1D}$ for $n> 2$ can also be computed similarly.
Numerical data for $G_{(n)}^{1D}$ can be easily computed using a computer,
and more examples will be listed in Appendix \ref{sec-app-1d-data}.
\end{Example}

Now let $W^{1D}$ be the total $n$-point function defined by:
\be
W^{1D}(\bm z)
=\sum_{n\geq 1} \sum_{j_1,\cdots,j_n \geq 1}
\frac{\pd^n F^{1D}(\bm t)}{\pd x_{j_1}\cdots \pd x_{j_n}}\bigg|_{\bm t=0} \cdot
z_1^{-j_1-1}\cdots z_n^{-j_n-1},
\ee
then one has:
\be
\begin{split}
W^{1D}(\bm z) :=&\sum_{(g,n)\not= (0,1)} W_{g,n}^{1D}(z_1,\cdots,z_n) \\
=& \sum_{g\geq 0, n\geq 1} W_{g,n}^{\text{fat}}(z_1,\cdots,z_n)\\
=& \sum_{n\geq 1} G_{(n)}^{1D} (z_1,\cdots,z_n).
\end{split}
\ee
This is simply a reassembling of all the correlators of the topological 1D gravity.
The free energy $F^{1D}(\bm t)$ can be recovered from $W^{1D}(\bm z)$ by:
\be
\label{eq-npt-fe}
F^{1D}(\bm t) = \Psi (W^{1D}(\bm z)),
\ee
where $\Psi $ is the following linear map between vector spaces:
\be
\begin{split}
\Psi :\quad  \prod_{n\geq 1} z_1^{-1}\cdots z_n^{-1} \bC[z_1^{-1},\cdots, z_n^{-1}]
&\quad\to\quad \bC[[\bm t]],\\
z_1^{-j_1-1}\cdots z_n^{-j_n-1} &\quad\mapsto\quad
\frac{x_{j_1}\cdots x_{j_n}}{ n! }
=\frac{t_{j_1-1}\cdots t_{j_n-1}}{n!\cdot j_1!\cdots j_n!}.
\end{split}
\ee

\section{Specialize to the Problem of Computing $\chi(\Mbar_{g,n})$}
\label{sec-compute-chi}

Now let us apply the above results of the topological 1D gravity
to the problem of the computation of the orbifold Euler characteristics of $\Mbar_{g,n}$.

\subsection{Orbifold characteristics of $\cM_{g,n}$ and $\Mbar_{g,n}$}

First let us recall some basic results on the orbifold Euler characteristics
$\chi(\cM_{g,n})$ and $\chi(\Mbar_{g,n})$.

Let $\cM_{g,n}$ be the moduli space of smooth stable curves of genus $g$ with $n$ marked points.
The orbifold Euler characteristics of $\cM_{g,n}$ is given by the famous Harer-Zagier formula
(see \cite{hz, pe}):
\be
\chi(\cM_{g, n})=(-1)^{n} \cdot \frac{(2 g-1) B_{2 g}}{(2 g) !}(2 g+n-3) !, \quad 2 g-2+n>0,
\ee
where $B_{2g}$ is the $(2g)$-th Bernoulli number.
Let us denote by $V_n(z)$ the following generating series of $\chi(\cM_{g,n})$ for every fixed $n$:
\be
\label{eq-def-gen-V}
\begin{split}
&V_{0}(z):=\sum_{g=2}^{\infty} \chi(\mathcal{M}_{g, 0}) z^{2-2 g} ; \\
&V_{n}(z):=\sum_{g=1}^{\infty} \chi(\mathcal{M}_{g, n}) z^{2-2 g-n}, \qquad n=1,2 ; \\
&V_{n}(z):=\sum_{g=0}^{\infty} \chi(\mathcal{M}_{g, n}) z^{2-2 g-n}, \qquad n \geq 3.
\end{split}
\ee
The series $V_{n}(z)$ can be related to the Gamma function as follows
(\cite[\S 4]{wz2}):
\begin{equation*}
\begin{split}
&V_0(z)
\sim \int_0^z \biggl(z\frac{d}{dz}\log\Gamma(z+1)\biggr) dz
-\half z^2\log z+\frac{1}{4}z^2 - \half z - C +\frac{1}{12}\log z,\\
&V_1(z)
\sim z\frac{d}{dz}\log\Gamma(z)-z\log z+\half,\\
&V_2(z)
\sim z\frac{d^2}{dz^2}\log\Gamma(z)+\frac{d}{dz}\log\Gamma(z)-\log z-1,\\
&V_n(z)
\sim z\frac{d^n}{dz^n}\log\Gamma(z)+(n-1)\frac{d^{n-1}}{dz^{n-1}}\log\Gamma(z),
\qquad n\geq 3;
\end{split}
\end{equation*}
and related to the Barnes $G$-function as follows
(\cite[\S 4]{wz2}):
\begin{equation*}
\begin{split}
&V_0(z) \sim \log G(z+1)
-\biggl( \zeta'(-1) + \frac{z}{2}\log (2\pi) + \biggl(\frac{z^2}{2}-\frac{1}{12}\biggr)\log z
- \frac{3z^2}{4} \biggr),\\
&V_1(z) \sim
\frac{d}{dz}\log G(z+1)
-\frac{1}{2} \log(2\pi) -z\log z +z,\\
&V_2(z) \sim \frac{d^2}{dz^2} \log G(z+1)-\log z,\\
&V_n (z) \sim \frac{d^n}{dz^n}\log  G(z+1),
\qquad  n\geq 3.
\end{split}
\end{equation*}

Let $\Mbar_{g,n}$ be the Deligne-Mumford compactification of $\cM_{g,n}$.
In \cite{bh},
Bini and Harer gave the following formula for the orbifold Euler characteristics
of $\Mbar_{g,n}$ as a summation over all connected stable graphs of genus $g$ with $n$ external edges
(here we denote bt $\cG_{g,n}^c$ the set of all such graphs):
\be
\label{eq-chi-graph}
\chi(\Mbar_{g, n})=n!\cdot
\sum_{\Gamma \in \cG_{g, n}^{c}}
\frac{1}{|\Aut(\Gamma)|} \prod_{v \in V(\Gamma)}
\chi(\cM_{g_{v}, \val_{v}}),
\ee
where $g_v$ is the genus of a vertex $v$, and $\val_v$ is the valence of $v$.
A consequence of this graph sum formula is the following (see \cite[(11)]{bh} and \cite[\S 6]{wz2}):
\be
\label{eq-generating-chibar}
\begin{split}
& \exp \bigg(\sum_{2 g-2+n>0} \frac{1}{n !} \chi(\Mbar_{g, n}) y^{n} z^{2-2 g}\bigg) \\
=& \frac{1}{\sqrt{2 \pi}} \int \exp \bigg(-\frac{1}{2}(x-y z)^{2}
+\sum_{n \geq 0} V_{n}(z) \cdot \frac{x^{n}}{n !}\bigg) d x.
\end{split}
\ee
However,
if one wants to compute the generating series
\ben
\sum_{2 g-2+n>0} \frac{1}{n !} \chi(\Mbar_{g, n}) y^{n} z^{2-2 g},
\een
expanding the above formal integral directly and then taking its logarithm is not an efficient way.

Since listing all possible graphs in the formula \eqref{eq-chi-graph} or
expanding the logarithm of the above formal integral are both unpractical in calculations,
one desires to find other tools to compute $\chi(\Mbar_{g,n})$.
In the earlier work \cite{wz2},
the authors have derived two types recursion relations to compute $\chi(\Mbar_{g, n})$
using the formalism of abstract QFT for stable graphs.
Using these recursions,
one is able to carry out the computations of all numerical data,
and derive some closed formulas.
Moreover,
we have related the computation of $\chi(\Mbar_{g,n})$ to the topological 1D gravity and KP hierarchy in
\cite[\S 6]{wz2}.
Then the results we discussed in last two sections can be applied to the problem,
and we will explain this in the following subsections.

\subsection{Orbifold characteristics $\Mbar_{g,n}$ and KP hierarchy}

Now let us recall the main results in \cite[\S 6]{wz2}.

Comparing the formula \eqref{eq-generating-chibar} to
the partition function \eqref{eq-1d-partitionZ-t} of the topological 1D gravity,
one can easily finds that:
\begin{Theorem}
[\cite{wz2}]
\label{thm-chi-KP}
Let $y,z$ be two formal variable and denote by $\chi(y,z)$ the following generating series
of the orbifold Euler characteristics of $\Mbar_{g,n}$:
\be
\label{eq-gen-chibar}
\chi(y,z):=
\sum_{2g-2+n>0} \frac{y^n z^{2-2g}}{n!} \cdot \chi(\Mbar_{g,n})
-\widetilde V_0(y,z),
\ee
then we have:
\be
\chi(y,z) = F^{1D}(\bm t)|_{t_n = \widetilde V_{n+1}(y,z),n\geq 0} ,
\ee
where $F^{1D}$ is the logarithm of the KP tau-function $Z^{1D}|_{\lambda=1}$ specified by the topological 1D gravity,
and $\widetilde V_n(y,z)$ are the following formal series:
\be
\label{eq-tildeV}
\widetilde V_n(y,z):=
-\half y^2 z^2\cdot\delta_{n,0} + yz\cdot \delta_{n,1} +V_n(z), \qquad n\geq 0,
\ee
and $V_n(z)$ are the generating series of $\chi(\cM_{g,n})$ given by \eqref{eq-def-gen-V}.
\end{Theorem}

As a corollary,
by taking $y=0$ one has:
\begin{Corollary}
[\cite{wz2}]
\label{cor-chi0-KP}
Let $\chi_0(z)$ be the following generating series of the orbifold Euler characteristics of $\Mbar_{g,0}$:
\be
\label{eq-gen-chibar-0}
\chi_0(z):=
\sum_{g\geq 2} \big( \chi(\Mbar_{g,0}) - \chi(\cM_{g,0}) \big)\cdot z^{2-2g}
\ee
Then we have:
\be
\chi_0 (z)= F^{1D}(\bm t)|_{t_n = V_{n+1}(y,z),n\geq 0} ,
\ee
where $F^{1D}$ is the logarithm of the KP tau-function $Z^{1D}$ specified by the topological 1D gravity,
and $V_n(z)$ are the generating series of $\chi(\cM_{g,n})$ given by \eqref{eq-def-gen-V}.
\end{Corollary}

\subsection{Computation of the generating series of $\chi(\Mbar_{g,n})$ using topological recursion}

Now we can use the results for the tau-function $Z^{1D}|_{\lambda=1}$ of the KP hierarchy
developed in \S \ref{sec-mainresult}
to compute the generating series $\chi(y,z)$ of $\chi(\Mbar_{g,n})$.

Recall that the free energy $F^{1D}=\log Z^{1D}$ can be recovered from the
total $n$-point function $W^{1D}$
using the formula \eqref{eq-npt-fe}.
Now let $\widetilde \Psi$ be the evaluation of the map $\Psi$ to the time:
\ben
t_n = \widetilde V_{n+1} (y,z),
\qquad n\geq 0,
\een
i.e., $\widetilde \Psi$ is the linear map:
\be
\label{eq-specialization-chi}
\begin{split}
\widetilde\Psi :\quad
\prod_{n\geq 1} z_1^{-1}\cdots z_n^{-1} \bC[z_1^{-1},\cdots, z_n^{-1}]
&\quad\to\quad \bC[[y]]((z)),\\
z_1^{-j_1-1}\cdots z_n^{-j_n-1} &\quad\mapsto\quad
\frac{\widetilde V_{j_1}(y,z)\cdots \widetilde V_{j_n}(y,z)}{n!\cdot j_1!\cdots j_n!},
\end{split}
\ee
where $\widetilde V_n(y,z)$ are given by \eqref{eq-tildeV}.
Then the generating series $\chi(y,z)$ defined by \eqref{eq-gen-chibar} can be computed using the following formula:
\begin{Theorem}
\label{thm-chi-1dnpt}
We have:
\be
\label{eq-chi-1dnpt}
\begin{split}
\chi(y,z) =& \widetilde \Psi \bigg(
\sum_{(g,n)\not=(0,1)} W_{g,n}^{1D} (z_1,\cdots, z_n)
\bigg)\\
=& \widetilde \Psi \bigg(
\sum_{g\geq 0,n\geq 1} W_{g,n}^{\text{fat}} (z_1,\cdots, z_n)
\bigg)
\end{split}
\ee
where $W_{g,n}^{1D}$ can be computed recursively using \eqref{eq-toporec-1}
together with the initial value $W_{0,1}^{1D}(z_1)=\frac{1}{z_1}$;
and $W_{g,n}^{\text{fat}}$ can be computed recursively using the E-O topological recursion \eqref{eq-EO-fat}.
\end{Theorem}

Similarly,
define $\widetilde\Psi_0$ to be the specialization:
\be
\label{eq-specialization-chi-0}
\begin{split}
\widetilde\Psi_0 :\quad
\prod_{n\geq 1} z_1^{-1}\cdots z_n^{-1} \bC[z_1^{-1},\cdots, z_n^{-1}]
&\quad\to\quad \bC[[z^{-1}]],\\
z_1^{-j_1-1}\cdots z_n^{-j_n-1} &\quad\mapsto\quad
\frac{ V_{j_1}(z)\cdots  V_{j_n}(z)}{n!\cdot j_1!\cdots j_n!},
\end{split}
\ee
where $V_n(z)=\widetilde V(0,z)$ are given by \eqref{eq-def-gen-V},
then the following is obtained by taking $y=0$ in \eqref{eq-chi-1dnpt}:
\begin{Corollary}
The generating series \eqref{eq-gen-chibar-0} of $\chi(\Mbar_{g,0})$ is given by:
\be
\chi_0(z) = \widetilde \Psi_0 \bigg(
\sum_{(g,n)\not=(0,1)} W_{g,n}^{1D} (z_1,\cdots, z_n)
\bigg)
=\widetilde \Psi_0 \bigg(
\sum_{g\geq 0,n\geq 1} W_{g,n}^{\text{fat}} (z_1,\cdots, z_n)
\bigg).
\ee
\end{Corollary}

\subsection{A formula for the generating series in terms of affine coordinates}
\label{sec-chi-affine}

One can rewrite the above formula for the generating series $\chi(y,z)$
in terms of the affine coordinates for the tau-function $Z^{1D}$.
Using the fact
\begin{equation*}
W^{1D}(\bm z) = \sum_{n\geq 1} G_{(n)}^{1D} (z_1,\cdots,z_n).
\end{equation*}
and Proposition \ref{prop-thm-npt-1d} which gives a formula for $G_{(n)}^{1D}$
in terms of the affine coordinates of $Z^{1D}|_{\lambda=1}$,
we easily have:
\begin{Theorem}
The generating series
\ben
\chi(y,z):=
\sum_{2g-2+n>0} \frac{1}{n!} y^n z^{2-2g}\cdot \chi(\Mbar_{g,n})
-\widetilde V_0(y,z)
\een
is given by the formula:
\begin{equation}
\label{eq-formula-chi}
\chi(y,z)
=\widetilde \Psi\bigg(\sum_{n\geq 1} (-1)^{n-1} \sum_{n\text{-cycles }\sigma}
\prod_{i=1}^{n}
\widehat A^{1D} (z_{\sigma(i)}, z_{\sigma(i+1)})
-\frac{1}{(z_1-z_2)^2}
\bigg),
\end{equation}
where $\widetilde\Psi$ is given by \eqref{eq-specialization-chi}, and
\be
\widehat A^{1D}(\xi,\eta)
=\sum_{k=0}^\infty (2k+1)!!\cdot \xi^{-1}\eta^{-2k-2}
+ \frac{1}{\xi-\eta}.
\ee
\end{Theorem}

Taking $y=0$, we have:
\begin{Corollary}
The generating series \eqref{eq-gen-chibar-0} is given by:
\begin{equation}
\chi_0(z)
=\widetilde \Psi_0 \bigg(\sum_{n\geq 1} (-1)^{n-1} \sum_{n\text{-cycles }\sigma}
\prod_{i=1}^{n}
\widehat A^{1D} (z_{\sigma(i)}, z_{\sigma(i+1)})
-\frac{1}{(z_1-z_2)^2}
\bigg),
\end{equation}
where the map $\widetilde\Psi_0$ is given by \eqref{eq-specialization-chi-0}.
\end{Corollary}

\subsection{A remark on the formulas}

Here let us remark that for every fixed pair $(g,n)$ with $2g-2+n>0$,
the coefficients of $y^n z^{2-2g}$ in the right-hand sides of \eqref{eq-chi-1dnpt} and \eqref{eq-formula-chi}
are summations of finite numbers of terms,
hence these formulas indeed work for practical computations even though they involve infinite summations.

In fact,
notice that the generating series $\widetilde V_n(y,z)$ ($n\geq 1$) are of the form:
\be
\begin{split}
\label{eq-data-tildeV}
& \widetilde V_1(y,z) =
yz -\frac{1}{12}z^{-1} +\frac{1}{120} z^{-3}-\frac{1}{252}z^{-5}
+ \frac{1}{240}z^{-7} -\frac{1}{132} z^{-9} +\cdots,\\
& \widetilde V_2(y,z) =
\frac{1}{12}z^{-2} -\frac{1}{40}z^{-4} +\frac{5}{252}z^{-6}
-\frac{7}{240}z^{-8} +\frac{3}{44}z^{-10}
-\cdots,\\
& \widetilde V_3(y,z) = z^{-1}
-\frac{1}{6}z^{-3}+ \frac{1}{10}z^{-5} -\frac{5}{42}z^{-7}+ \frac{7}{30}z^{-9} -\frac{15}{22}z^{-11} +\cdots,\\
& \widetilde V_4(y,z) = -z^{-2}
+\frac{1}{2}z^{-4} -\frac{1}{2}z^{-6} +\frac{5}{6}z^{-8} -\frac{21}{10}z^{-10} +\frac{15}{2}z^{-12} -\cdots,\\
& \widetilde V_5(y,z) = 2z^{-3}
-2 z^{-5} +3z^{-7} -\frac{20}{3}z^{-9}+ 21z^{-11} -90z^{-13} +\cdots,\\
&\cdots\cdots
\end{split}
\ee
In particular, when $z\to\infty$, one has:
\be
\label{eq-property-z}
\begin{split}
&\widetilde V_1(y,z) = yz + O(z^{-1}),\\
&\widetilde V_2(y,z) =  O(z^{-2}),\\
&\widetilde V_n(y,z) =  O(z^{-(n-2)}),\qquad n\geq 3.
\end{split}
\ee
Now fix a pair $(g,n)$ with $2g-2+n>0$,
and consider all the terms which are nonzero multiples of $y^n z^{2-2g}$
in the right-hand side of \eqref{eq-chi-1dnpt} or \eqref{eq-formula-chi}.
We claim:
\begin{Lemma}
Let $p_1,p_2,\cdots,p_m$ be some nonnegative integers.
If the coefficient of $y^n z^{2-2g}$ is nonzero in
$\widetilde V_1^{p_1}(y,z) \widetilde V_2^{p_2}(y,z) \cdots \widetilde V_m^{p_m}(y,z)$,
then:
\be
\label{ineq-finiteness}
p_1+ 2p_2+\sum_{i=3}^m (i-2) p_i \leq 2n-2+2g.
\ee
\end{Lemma}
\begin{proof}
We have:
\ben
\widetilde V_1^{p_1} \widetilde V_2^{p_2} \cdots \widetilde V_m^{p_m}
= \sum_{l=0}^{p_1}
\binom{k}{l}\cdot (yz)^l \cdot
(-\frac{1}{12}z^{-1} +\cdots)^{p_1-l}
\cdot V_2^{p_2} \cdots \widetilde V_m^{p_m}.
\een
If there is a nonzero multiple of $y^n z^{2-2g}$ in
$\widetilde V_1^{p_1} \widetilde V_2^{p_2} \cdots \widetilde V_m^{p_m}$,
by \eqref{eq-property-z} one must have:
\ben
2-2g\leq n-(p_1-n) - 2p_2 - p_3 -2p_4 - \cdots -(m-2)p_m.
\een
This proves the inequality \eqref{ineq-finiteness}.
\end{proof}

Now given a pair $(g,n)$ with $2g-2+n>0$,
assume that $\widetilde\Psi (z_1^{-j_1-1}\cdots z_N^{-j_N-1})$ contains a nonzero multiple of $y^nz^{2-2g}$
where $\widetilde \Psi$ is the specialization \eqref{eq-specialization-chi}.
Then using the condition \eqref{ineq-finiteness} one easily sees that $N\leq 2n-2+2g$ and
$j_i\leq 2n+2g$ for every $i=1,2,\cdots, N$.
This tells us that for fixed $(g,n)$ we only need to compute a finite number of terms in the $n$-point functions
in the formulas we have derived in last two subsections to get the coefficient of $y^nz^{2-2g}$.
Hence one can indeed compute the numbers $\chi(\Mbar_{g,n})$ using these formulas
with the help of a compute.

\begin{Example}
Let us give some simple examples of the computations.
Consider $(g,n)=(1,1)$,
and we need to find all the terms $z_1^{-j_1-1}\cdots z_N^{-j_N-1}$ in the total $n$-point functions
with $N\leq 2$ and $j_i\leq 4$.
Using the data listed in Appendix \ref{sec-app-1d-data},
one only needs to compute the coefficient of $y^1 z^0 $ in:
\ben
\half \widetilde V_2 +\frac{1}{8}\widetilde V_4
+\half \widetilde V_1^2 +\frac{1}{2}\widetilde V_1 \widetilde V_3
+\frac{1}{4} \widetilde V_2 \widetilde V_4
+\frac{1}{4} \widetilde V_2^2 +\frac{5}{24} \widetilde V_3^2 + \frac{1}{12}\widetilde V_4^2.
\een
One easily sees that only $\half \widetilde V_1^2$ and $\frac{1}{2}\widetilde V_1 \widetilde V_3$
have nontrivial contributes since all the others are independent of $y$.
Then using the data \eqref{eq-data-tildeV} one finds that
the coefficient of $y^1 z^0$ in $\half \widetilde V_1^2 + \frac{1}{2}\widetilde V_1 \widetilde V_3$ is:
\ben
\half \cdot 2\cdot 1 \cdot(-\frac{1}{12}) + \half \cdot 1 \cdot 1=\frac{5}{12}.
\een
This tells that $\chi(\Mbar_{1,1})= \frac{5}{12}$.

Similarly
consider the case $(g,n)=(1,2)$,
and all the terms that have nonzero contribute to the coefficient of $y^2 z^0$ are:
\ben
\half \widetilde V_1^2\widetilde V_2 +\half \widetilde V_1^2\widetilde V_3^2
+\frac{1}{4}\widetilde V_1^2\widetilde V_4 +\frac{1}{6}\widetilde V_1^3\widetilde V_3,
\een
and thus the coefficient of $y^2 z^0$ is:
\ben
\half\cdot \frac{1}{12} + \half \cdot 1^2 +
\frac{1}{4}\cdot (-1)+ \frac{1}{6}\cdot 3\cdot (-\frac{1}{12})\cdot 1 =
\frac{1}{4},
\een
thus $\chi(\Mbar_{1,2})=2!\cdot \frac{1}{4} = \half$.

More numerical data of $\chi(\Mbar_{g,n})$ will be given in Appendix \ref{sec-app-data-chi}.
\end{Example}

\vspace{.2in}


\begin{appendices}

\section{A Combinatorial Calculation of the Affine Coordinates of $Z^{1D}$}
\label{sec-proof-comb}

In this appendix,
we give a proof to Theorem \ref{thm-affinecoord-1D}
by a simple observation concerning symmetric functions and ordinary graphs.

\subsection{Preliminaries of the boson-fermion correspondence}
\label{sec-app-pre}

First let us recall some basic knowledge of the boson-fermion correspondence.
We only state some necessary facts without proof,
for details see the book \cite{djm}.
Here the notations we use are following \cite{zhou1}.

Let $U\in Gr_{(0)}$ be an element of the big cell of the Sato Grassmannian.
It is a subspace of $H= \bC[z]\oplus z^{-1}\bC [[z^{-1}]]$ spanned by the normalized basis:
\be
\big\{
f_n
=z^{n} + \sum_{m\geq 0} a_{n,m} z^{-m-1}
\big\}_{n\geq 0},
\ee
where $\{a_{n,m}\}_{n,m\geq 0}$ are the affine coordinates for $U$.
This basis determines a semi-infinite wedge product:
\begin{equation*}
\begin{split}
(z^\half f_0)\wedge (z^\half f_1)\wedge (z^\half f_2)\wedge \cdots
=& \sum \alpha_{m_1,\cdots, m_k;n_1,\cdots, n_k}
z^{-m_1-\half}\wedge\cdots\wedge z^{-m_k-\half} \\
&\quad \wedge z^{\half}\wedge z^{\frac{3}{2}}\wedge\cdots\wedge \widehat{z^{n_k+\half}}
\wedge\cdots \wedge \widehat{z^{n_1+\half}} \wedge \cdots,
\end{split}
\end{equation*}
where $m_1>m_2>\cdots >m_k\geq 0$ and $n_1>n_2>\cdots >n_k\geq 0$ are two sequences of integers,
and the coefficients are given by:
\ben
\alpha_{m_1,\cdots, m_k;n_1,\cdots, n_k}
=(-1)^{n_1+\cdots+n_k} \cdot
\det
\begin{pmatrix}
a_{n_1,m_1} & \cdots & a_{n_1,m_k} \\
\vdots & \vdots & \vdots\\
a_{n_k,m_1} & \cdots & a_{n_k,m_k}
\end{pmatrix}.
\een
Thus we get a linear map:
\be
Gr_{(0)}\to \cF^{(0)},\qquad
U=\text{span}\{f_0,f_1,f_2\cdots\}
\mapsto |U\rangle = (z^\half f_0)\wedge (z^\half f_1) \wedge \cdots,
\ee
where $\cF^{(0)}$ is the fermionic Fock space of charge $0$
(see eg. \cite[\S 3]{zhou1} for definitions and notations),
and $\{f_n\}_{n\geq 0}$ is the normalized basis for $U$.
Moreover, one has:
\begin{Theorem}
[\cite{zhou1}]
\label{thm-coeff-Bogoliubov}
Let $\{a_{n,m}\}_{n,m\geq 0}$ be the affine coordinates of $U\in Gr_{(0)}$.
Then $|U\rangle \in \cF^{(0)}$ is equal to the following Bogoliubov transform of the fermionic vacuum $|0\rangle \in \cF^{(0)}$:
\ben
|U\rangle = e^A |0\rangle,
\een
where $A:\cF^{(0)}\to\cF^{(0)}$ is defined by:
\be
\label{eq-def-bog}
A:= \sum_{n,m\geq 0} a_{n,m} \psi_{-m-\half} \psi_{-n-\half}^*,
\ee
and $\psi_{-m-\half}, \psi_{-n-\half}^*$ are the fermionic creators for $m,n\geq 0$.
\end{Theorem}

The bosonic Fock space $\cB$ is defined to be
$\cB:=\Lambda[w,w^{-1}]$,
where $\Lambda$ is the space of symmetric functions,
and $w$ is a formal variable.
The boson-fermion correspondence is a certain linear isomorphism
\ben
\Phi:\quad\cF =\bigoplus_{n\in\bZ}\cF^{(n)} \quad \to\quad \cB
\een
of vector spaces.
In particular,
$\Phi |_{\cF^{(0)}} :  \cF^{(0)} \to \Lambda$
is an isomorphism of vector spaces.
Moreover,
the bosonic operators acting on $\cB$ and the fermionic operators $\psi_r,\psi_s^*$ acting on $\cF$
can be represented in terms of each other.
See \cite{djm} for details.

There is a natural basis for $\cF^{(0)}$ labeled by partitions of non-negative integers.
Let $\mu=(\mu_1,\mu_2,\cdots)$ be a partition of a positive integer whose Frobenius notation is
$\mu=(m_1,\cdots, m_k | n_1,\cdots,n_k)$
(see eg. \cite{mac}).
Then $\{|\mu\rangle\}$ is a basis for $\cF^{(0)}$ where:
\be
|\mu\rangle=(-1)^{n_1+\cdots+n_k}\cdot
\psi_{-m_1-\half} \psi_{-n_1-\half}^* \cdots
\psi_{-m_k-\half} \psi_{-n_k-\half}^* |0\rangle.
\ee
Here for the zero partition $\mu=(0)$ we define $|(0)\rangle = |0\rangle$ to be the vacuum.

The space $\Lambda$ of symmetric functions also has a basis labeled by partitions.
Let $s_\mu$ be the Schur function (see eg. \cite{mac}) labeled by the partition $\mu$,
then $\{s_\mu\}$ is a basis for $\Lambda$.
An important property of the boson-fermion correspondence $\Phi : \cF^{(0)} \to \Lambda$
is that it takes $|\mu\rangle$ to $s_\mu$:
\ben
\Phi |_{\cF^{(0)}} : \qquad
 \cF^{(0)} \to \Lambda,
\qquad
|\mu\rangle \mapsto s_\mu=
\langle 0 | e^{\sum_{n=1}^\infty \frac{p_n}{n} \alpha_n} | \mu \rangle,
\een
where $p_n$ are Newton symmetric functions and $\alpha_n$ are some bosonic operators.
As a corollary,
one has:
\be
\Phi (|U\rangle) = 1+ \sum_{|\mu|>0} (-1)^{n_1+\cdots +n_k} \cdot
\det (a_{n_i,m_j})_{1\leq i,j \leq k} \cdot s_\mu,
\ee
where $\{a_{n,m}\}_{n,m\geq 0}$ are the affine coordinates for $U$,
and $(m_1,\cdots, m_k | n_1,\cdots,n_k)$ is the Frobenius notation of the partition $\mu$.

Sato \cite{sa} showed that given an element $U$ on the Sato Grassmannian,
the vector $\Phi(|U\rangle)$ in the bosonic Fock space is a tau-function of the KP hierarchy
with respect to time variables $T_n:= \frac{p_n}{n}$ ($n=1,2,\cdots$).
Moreover,
every (formal power series) tau-function can be constructed in this way,
thus the Sato Grassmannian is the space of all tau-functions of the KP hierarchy.
It is the orbit of an $\widehat{GL(\infty)}$-action on the trivial solution $\tau=1$.

\subsection{Symmetric functions and Feynman diagrams}

Now let us recall some backgrounds of symmetric functions.

Denote by $p_n$ the Newton symmetric function of degree $n$:
\ben
p_n (\bm x):= x_1^n+x_2^n+ x_3^n +\cdots,
\een
then $\{p_\lambda\}_{\lambda}$ is a basis of the space of symmetric functions $\Lambda$,
where $\lambda=(\lambda_1, \cdots, \lambda_k)$
(with $\lambda_1\geq \cdots\geq\lambda_k>0$)
is a partition of an integer
$|\lambda|:=\lambda_1+\cdots +\lambda_k$,
and
\be
p_\lambda:=p_{\lambda_1}p_{\lambda_2}\cdots p_{\lambda_k}.
\ee
We also denote $p_{(0)}:=1$ for the trivial partition $(0)$ of $0$.

Let $h_n\in\Lambda$ be the
complete symmetric function of degree $n$:
\ben
h_n (\bm x):= \sum_{\sum_i d_i =n} x_1^{d_1}x_2^{d_2}x_3^{d_3}\cdots.
\een
Then it can be spanned with respect to the basis $\{p_\lambda\}_{\lambda}$ by
the following well-known relation
(see eg. \cite[\S 1]{mac}):
\be
\label{eq-relation-symm}
h_n = \sum_{|\lambda|=n} \frac{1}{z_\lambda} p_\lambda,
\ee
where
\be
z_\lambda:=\prod_{i\geq 1} i^{m_i}\cdot m_i !,
\ee
and $m_i$ is the number of $i$ in the partition $\lambda$,
i.e.,
$\lambda=(1^{m_1} 2^{m_2} 3^{m_3}\cdots)$.

An easy but interesting observation is that the above formula \eqref{eq-relation-symm} for even $n$
can be interpreted in terms of Feynman graphs:
\begin{Proposition}
For every $n\geq 1$,
the complete symmetric function $h_{2n}$ can be represented as
a summation over (not necessarily stable, not necessarily connected) ordinary graphs:
\be
\label{eq-twosymm-2}
h_{2n}= \frac{1}{(2n-1)!!}\cdot
\sum_{\text{$\Gamma:$ $|E(\Gamma)|=n$}}
\frac{1}{|\Aut(\Gamma)|} p_\Gamma,
\ee
where $|E(\Gamma)|$ is the set of edges of $\Gamma$,
and
\be
p_\Gamma:=
\prod_{v:\text{ vertex}}(\val(v)-1)!\cdot p_{\val(v)}.
\ee
\end{Proposition}
\begin{proof}
Given a partition $\lambda=(\lambda_1,\lambda_2,\cdots,\lambda_k)$ of $|\lambda|=2n$,
let us denote by $\cG_\lambda$ the set of all ordinary graphs
(not necessarily stable, not necessarily connected)
whose vertices are of valences $\lambda_1,\lambda_2,\cdots,\lambda_k$ respectively.
We only need to prove:
\be
h_{2n} = \frac{1}{(2n-1)!!}\cdot
\sum_{|\lambda|=2n} \sum_{\Gamma\in \cG_\lambda}
\frac{\prod\limits_{v:\text{ vertex}}(\val(v)-1)!}{|\Aut(\Gamma)|}p_\lambda.
\ee
And by \eqref{eq-relation-symm},
it suffices to show that:
\be
\frac{(\lambda_1-1)!\cdots (\lambda_k-1)!}{(2n-1)!!} \cdot \sum_{\Gamma\in \cG_\lambda}
\frac{1}{|\Aut(\Gamma)|}
= \frac{1}{z_\lambda},
\ee
for every partition $\lambda=(\lambda_1,\cdots,\lambda_k)$ with $|\lambda|=2n$.
Recall that
\ben
z_\lambda= \prod_{i\geq 1} i^{m_i}\cdot m_i !
= \lambda_1\cdots\lambda_k \cdot \prod_{i\geq 1}  m_i !,
\een
where $m_i$ is the number of $i$ appearing in $\lambda$.
Thus we only need to show:
\be
\sum_{\Gamma\in \cG_\lambda}
\frac{1}{|\Aut(\Gamma)|}
= \frac{(2n-1)!!}{\lambda_1!\cdots \lambda_k!\cdot\prod_{i\geq 1}  m_i !},
\ee
and this is a well-known result in graph-counting,
see eg. \eqref{eq-graphcounting}.
\end{proof}

\begin{Example}
We give some examples of the above proposition.
First consider $k=1$.
In this case,
there are only two graphs with one edge:
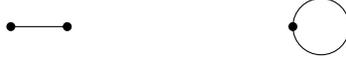
\begin{equation*}
\begin{tikzpicture}[scale=1.5]
\draw [fill] (0,0) circle [radius=0.035];
\draw [fill] (0.5,0) circle [radius=0.035];
\draw (0,0) -- (0.5,0);
\draw [fill] (2.5,0) circle [radius=0.035];
\draw (2.75,0) circle [radius=0.25];
\end{tikzpicture}
\end{equation*}
The automorphism groups of them are both of order $2$.
Then the formula \eqref{eq-twosymm-2} gives:
\ben
h_2  = \half p_1^2 + \half p_2.
\een
Now consider $k=2$.
In this case,
there are $7$ graphs with two edges,
and four of them are connected:
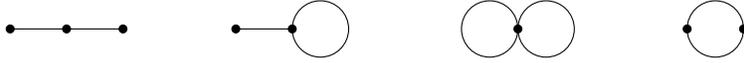
\begin{equation*}
\begin{tikzpicture}[scale=1.5]
\draw [fill] (0,0) circle [radius=0.035];
\draw [fill] (0+0.5,0) circle [radius=0.035];
\draw [fill] (0.5+0.5,0) circle [radius=0.035];
\draw (0,0) -- (0.5+0.5,0);
\draw [fill] (2,0) circle [radius=0.035];
\draw [fill] (2.5,0) circle [radius=0.035];
\draw (2.75,0) circle [radius=0.25];
\draw (2,0) -- (2.5,0);
\draw [fill] (4.5,0) circle [radius=0.035];
\draw (4.75,0) circle [radius=0.25];
\draw (4.25,0) circle [radius=0.25];
\draw [fill] (6.5,0) circle [radius=0.035];
\draw [fill] (6,0) circle [radius=0.035];
\draw (6.25,0) circle [radius=0.25];
\end{tikzpicture}
\end{equation*}
and the other three are disconnected:
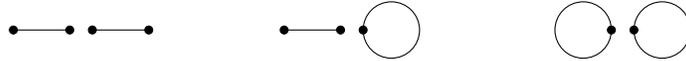
\begin{equation*}
\begin{tikzpicture}[scale=1.5]
\draw [fill] (-0.2-0.3,0) circle [radius=0.035];
\draw [fill] (0.3-0.3,0) circle [radius=0.035];
\draw [fill] (0+0.5-0.3,0) circle [radius=0.035];
\draw [fill] (0.5+0.5-0.3,0) circle [radius=0.035];
\draw (-0.2-0.3,0) -- (0.3-0.3,0);
\draw (0.5-0.3,0) -- (1-0.3,0);
\draw [fill] (1.9,0) circle [radius=0.035];
\draw [fill] (2.4,0) circle [radius=0.035];
\draw [fill] (2.6,0) circle [radius=0.035];
\draw (2.85,0) circle [radius=0.25];
\draw (1.9,0) -- (2.4,0);
\draw [fill] (4.5+0.3,0) circle [radius=0.035];
\draw [fill] (4.7+0.3,0) circle [radius=0.035];
\draw (4.95+0.3,0) circle [radius=0.25];
\draw (4.25+0.3,0) circle [radius=0.25];
\end{tikzpicture}
\end{equation*}
Then by \eqref{eq-twosymm-2} we have:
\begin{equation*}
\begin{split}
h_4 =& \frac{1}{3}\bigg(
\half p_1^2 p_2 + \frac{1}{2} p_1 (2!\cdot p_3)
+ \frac{1}{8} (3!\cdot p_4) + \frac{1}{4} p_2^2
+ \frac{1}{8} p_1^4 + \frac{1}{4} p_1^2 p_2
+ \frac{1}{8} p_2^2 \bigg)\\
=& \frac{1}{24} p_1^4 + \frac{1}{4} p_1^2 p_2 +\frac{1}{8} p_2^2
+ \frac{1}{3} p_1 p_3 +\frac{1}{4} p_4 .
\end{split}
\end{equation*}
For $k=3$ and $k=4$,
we omit the graphs and only write down the result:
\begin{equation*}
\begin{split}
h_6 = &\frac{1}{720} p_1^6 +\frac{1}{48} p_1^4p_2 +\frac{1}{16}p_1^2p_2^2
 +\frac{1}{18} p_1^3 p_3 +\frac{1}{48} p_2^3 +\frac{1}{6} p_1p_2p_3
  +\frac{1}{8}p_1^2p_4\\
& +\frac{1}{5}p_1p_5 +\frac{1}{8}p_2p_4
 +\frac{1}{18}p_3^2 +\frac{1}{6}p_6,\\
h_8 = & \frac{1}{40320} p_1^8 +\frac{1}{1440}p_1^4 p_2 +\frac{1}{360}p_1^3 p_3
+\frac{1}{192}p_1^4p_2^2 +\frac{1}{36}p_1^3p_2p_3 +\frac{1}{96}p_1^4p_4 \\
&+\frac{1}{96}p_1^2p_2^3 +\frac{1}{30}p_1^3p_5 +\frac{1}{36}p_1^2p_3^2
+\frac{1}{384}p_2^4 +\frac{1}{16}p_1^2p_2p_4 +\frac{1}{24}p_1p_2^2p_3\\
&+ \frac{1}{12}p_1^2p_6 +\frac{1}{32}p_2^2p_4 +\frac{1}{36}p_2p_3^2
+\frac{1}{10}p_1p_2p_5 +\frac{1}{12}p_1p_3p_4 +\frac{1}{32}p_4^2\\
&+ \frac{1}{12}p_2p_6 +\frac{1}{15}p_3p_5 +\frac{1}{7}p_1p_7 +\frac{1}{8}p_8.
\end{split}
\end{equation*}
\end{Example}

\begin{Remark}
On the ring $\Lambda$ of symmetric functions there is an involution $\omega:\Lambda \to \Lambda$ such that
\begin{align}
\omega(h_n)&= e_n, &\omega(e_n) &= h_n, & \omega(p_n) = (-1)^{n-1}p_n.
\end{align}
See \cite[\S I.2]{mac}.
By applying $\omega$ to the formulas above,
we can also interpret the formula of $e_n$ in terms of $p_k$ as Feynman sums.
\end{Remark}

\subsection{Representing $Z^{1D}$ in terms of symmetric functions}

Now recall the graph-sum formula \eqref{eq-1d-partitionZ-tFeynman} for the partition function $Z^{1D}$.
If we take $\lambda=1$ and
\be
t_n = n!\cdot p_{n+1}\in \Lambda,
\qquad \forall n\geq 0,
\ee
then:
\be
Z^{1D} (\bm t) = \sum_{\Gamma \in \cG^{or}} \frac{1}{|\Aut(\Gamma)|}
\prod_{v\in V(\Gamma)} (\val(v)-1)!\cdot p_{\val(v)}.
\ee
Applying \eqref{eq-twosymm-2},
we obtain:
\begin{Theorem}
\label{thm-Z-completesymm}
Taking $t_n = n!\cdot p_{n+1}\in \Lambda$ for every $n\geq 0$,
then we have:
\be
Z^{1D} (\bm t)=
1+\sum_{k=0} ^\infty (2k+1)!!\cdot h_{2k+2}.
\ee
\end{Theorem}

Now let us briefly recall the definition of Schur functions (see eg. \cite{mac}).
Let $\lambda=(m_1,\cdots,m_k|n_1,\cdots,n_k)$ be the Frobenius notation of a partition function $\lambda$,
then the Schur function labeled by $\lambda$ is defined by:
\be
s_\lambda := \det (s_{(m_i|n_j)})_{1\leq i,j\leq k},
\ee
where $s_{(m|n)}$ for a hook partition $\lambda=(m|n)=(m+1,1^n)$
is defined to be:
\be
\label{eq-def-Schur}
s_{(m|n)}= h_{m+1}e_n - h_{m+2}e_{n-1} + \cdots
+ (-1)^n h_{m+n+1}.
\ee
The set of all Schur functions $\{s_\lambda\}_\lambda$ forms another basis
for the space of symmetric functions $\Lambda$.
In particular,
the complete symmetric functions $h_n$ are some special Schur functions.
In fact,
taking $n=0$ in \eqref{eq-def-Schur}, then we see:
\ben
h_{m+1}=s_{(m|0)},\qquad
\forall m\geq 0,
\een
thus Theorem \ref{thm-Z-completesymm} can be rewritten as:
\begin{Corollary}
\label{cor-Z-schur}
Taking $t_n = n!\cdot p_{n+1}\in \Lambda$ for every $n\geq 0$,
then we have:
\be
\label{eq-Z-schur}
Z^{1D} (\bm t)=
1+\sum_{k=0} ^\infty (2k+1)!!\cdot s_{(2k+1|0)}.
\ee
\end{Corollary}

\subsection{Determination of the affine coordinates}

Now we are able to give a proof to Theorem \ref{thm-affinecoord-1D}.

Recall that under the boson-fermion correspondence \cite{djm}:
\ben
|a\rangle\quad
\mapsto \quad
\langle 0 | e^{\sum_{n\geq 1}\frac{p_n}{n}\alpha_n}
|a\rangle,
\een
the vector $|\mu\rangle$ is mapped to $s_\mu$,
where $\mu=(m_1,\cdots,m_k| n_1,\cdots,n_k)$ is the Frobenius notation
of a partition $\mu$.
Thus the formula \eqref{eq-Z-schur} simply tells us that in the fermionic picture,
the tau-function $Z^{1D}$ is given by:
\be
\label{eq-Z-ferm-pf}
\begin{split}
Z^{1D}=& |0\rangle +
\sum_{k=0}^\infty (2k+1)!! \cdot |(2k+1|0)\rangle\\
=& |0\rangle +
\sum_{k=0}^\infty (2k+1)!! \cdot
\psi_{-2k-\frac{3}{2}}\psi_{-\half}^* |0\rangle.
\end{split}
\ee
Recall that the fermionic creators $\psi_r,\psi_s^*$ ($r,s>0$) anti-commute with each other:
\begin{equation*}
\psi_r\psi_s + \psi_s\psi_r =
\psi_r^*\psi_s^* + \psi_s^*\psi_r^* =
\psi_r\psi_s^* + \psi_s^*\psi_r=0,
\qquad \forall r,s>0,
\end{equation*}
thus the above formula is equivalent to
$Z^{1D} = \exp (A^{1D}) |0\rangle$,
where:
\ben
A^{1D}=
\sum_{k=0}^\infty (2k+1)!! \cdot \psi_{-2k-\frac{3}{2}} \psi_{-\half}^*.
\een
Comparing this with Theorem \ref{thm-coeff-Bogoliubov},
we easily see that the affine coordinates for the tau-function $Z^{1D}$ is given by:
\ben
a_{n,m}^{1D} = \begin{cases}
m!!, & \text{if $n=0$ and $m$ is odd;}\\
0, & \text{otherwise.}
\end{cases}
\een
This proves Theorem \ref{thm-affinecoord-1D}.

\section{Some Numerical Data}
\label{sec-app-data}

In this appendix let us present more numerical data using the formulas derived in this work.

\subsection{$N$-point functions of the topological 1D gravity}
\label{sec-app-1d-data}

In this subsection we give some data of the $n$-point functions $G_{(n)}^{1D}(z_1,\cdots,z_n)$
of the topological 1D gravity.
They are computed using the formula \eqref{eq-thm-npt-1d}.

For simplicity let us denote by $\frac{1}{z^{(n_1,\cdots,n_k)}}$ the summation of all distinct terms
of the form $\frac{1}{z_{m_1}^{n_1}\cdots z_{m_k}^{n_k}}$ where $\{m_1,\cdots,m_k\}=\{1,2,\cdots,k\}$,
for example:
\footnotesize
\begin{equation*}
\begin{split}
&\frac{1}{z^{(2,2,2)}}:=\frac{1}{z_1^2z_2^2z_3^2},\\
&\frac{1}{z^{(2,2,5)}}:=
\frac{1}{z_1^2z_2^2z_3^5}+\frac{1}{z_1^2z_2^5z_3^2}+\frac{1}{z_1^5z_2^2z_3^2},\\
&\frac{1}{z^{(2,3,4)}}:=
\frac{1}{z_1^2z_2^3z_3^4} +\frac{1}{z_1^2z_2^4z_3^3}
+\frac{1}{z_1^3z_2^2z_3^4}
+\frac{1}{z_1^3z_2^4z_3^2} +\frac{1}{z_1^4z_2^2z_3^3} +\frac{1}{z_1^4z_2^3z_3^2}.
\end{split}
\end{equation*}
\normalsize
Then the following data are computed using \eqref{eq-thm-npt-1d}:
\footnotesize
\begin{flalign*}
\begin{split}
&G_{(2)}^{1D}(z_1,z_2)=
\frac{1}{z^{(2,2)}}
+ \frac{3}{z^{(2,4)}}
+\frac{2}{z^{(3,3)}}
+\frac{15}{z^{(2,6)}}
+\frac{12}{z^{(3,5)}}
+\frac{15}{z^{(4,4)}}
 +\frac{105}{z^{(2,8)}}
 +\frac{90}{z^{(3,7)}}
+\frac{105}{z^{(4,6)}}\\
&\quad
+\frac{96}{z^{(5,5)}}
+\frac{945}{z^{(2,10)}}
+\frac{840}{z^{(3,9)}}
 +\frac{945}{z^{(4,8)}}
+\frac{900}{z^{(5,7)}}
 +\frac{945}{z^{(6,6)}}
+\frac{10395}{z^{(2,12)}}
+\frac{9450}{z^{(3,11)}}
+\frac{10395}{z^{(4,10)}}
+\frac{10080}{z^{(5,9)}}\\
&\quad
+\frac{10395}{z^{(6,8)}}
+\frac{10170}{z^{(7,7)}}
+\frac{135135}{z^{(2,14)}}
+\frac{124740}{z^{(3,13)}}
+\frac{135135}{z^{(4,12)}}
+\frac{132300}{z^{(5,11)}}
+\frac{135135}{z^{(6,10)}}
+\frac{133560}{z^{(7,9)}}
+\frac{135135}{z^{(8,8)}}\\
&\quad
+\frac{2027025}{z^{(2,16)}}
+\frac{1891890}{z^{(3,15)}}
+\frac{2027025}{z^{(4,14)}}
+\frac{1995840}{z^{(5,13)}}
+\frac{2027025}{z^{(6,12)}}
+\frac{2012850}{z^{(7,11)}}
+\frac{2027025}{z^{(8,10)}}
+\frac{2016000}{z^{(9,9)}}\\
&\quad
+\frac{34459425}{z^{(2,18)}}
+\frac{32432400}{z^{(3,17)}}
+\frac{34459425}{z^{(4,16)}}
+\frac{34054020}{z^{(5,15)}}
+\frac{34459425}{z^{(6,14)}}
+\frac{34303500}{z^{(7,13)}}
+\frac{34459425}{z^{(8,12)}}\\
&\quad
+\frac{34360200}{z^{(9,11)}}
+\frac{34459425}{z^{(10,10)}}
+\frac{654729075}{z^{(2,20)}}
+\frac{620269650}{z^{(3,19)}}
+\frac{654729075}{z^{(4,18)}}
+\frac{648648000}{z^{(5,17)}}\\
&\quad
+\frac{654729075}{z^{(6,16)}}
+\frac{652702050}{z^{(7,15)}}
+\frac{654729075}{z^{(8,14)}}
+\frac{653637600}{z^{(9,13)}}
+\frac{654729075}{z^{(10,12)}}
+\frac{653836050}{z^{(11,11)}}
+\cdots.
\end{split}&&
\end{flalign*}
\begin{flalign*}
\begin{split}
&G_{(3)}^{1D}(z_1,z_2,z_3)=
\frac{2}{z^{(2,2,3)}}
+\frac{12}{z^{(2,2,5)}}
+\frac{12}{z^{(2,3,4)}}
+\frac{8}{z^{(3,3,3)}}
+\frac{90}{z^{(2,2,7)}}
+\frac{90}{z^{(2,3,6)}}
+\frac{96}{z^{(2,4,5)}}\\
&\quad
+\frac{72}{z^{(3,3,5)}}
+\frac{90}{z^{(3,4,4)}}
+\frac{840}{z^{(2,2,9)}}
+\frac{840}{z^{(2,3,8)}}
+\frac{900}{z^{(2,4,7)}}
+\frac{900}{z^{(2,5,6)}}
+\frac{720}{z^{(3,3,7)}}
+\frac{840}{z^{(3,4,6)}}\\
&\quad
+\frac{768}{z^{(3,5,5)}}
+\frac{900}{z^{(4,4,5)}}
+\frac{9450}{z^{(2,2,11)}}
+\frac{9450}{z^{(2,3,10)}}
+\frac{10080}{z^{(2,4,9)}}
+\frac{10080}{z^{(2,5,8)}}
+\frac{10170}{z^{(2,6,7)}}
+\frac{8400}{z^{(3,3,9)}}\\
&\quad
+\frac{9450}{z^{(3,4,8)}}
+\frac{9000}{z^{(3,5,7)}}
+\frac{9450}{z^{(3,6,6)}}
+\frac{10170}{z^{(4,4,7)}}
+\frac{10080}{z^{(4,5,6)}}
+\frac{9504}{z^{(5,5,5)}}
+\frac{124740}{z^{(2,2,13)}}
+\frac{124740}{z^{(2,3,12)}}\\
&\quad
+\frac{132300}{z^{(2,4,11)}}
+\frac{132300}{z^{(2,5,10)}}
+\frac{133560}{z^{(2,6,9)}}
+\frac{133560}{z^{(2,7,8)}}
+\frac{113400}{z^{(3,3,11)}}
+\frac{124740}{z^{(3,4,10)}}
+\frac{120960}{z^{(3,5,9)}}
+\frac{124740}{z^{(3,6,8)}}\\
&\quad
+\frac{122040}{z^{(3,7,7)}}
+\frac{133560}{z^{(4,4,9)}}
+\frac{132300}{z^{(4,5,8)}}
+\frac{133560}{z^{(4,6,7)}}
+\frac{128160}{z^{(5,5,7)}}
+\frac{132300}{z^{(5,6,6)}}
+\frac{1891890}{z^{(2,2,15)}}
+\frac{1891890}{z^{(2,3,14)}}\\
&\quad
+\frac{1995840}{z^{(2,4,13)}}
+\frac{1995840}{z^{(2,5,12)}}
+\frac{2012850}{z^{(2,6,11)}}
+\frac{2012850}{z^{(2,7,10)}}
+\frac{2016000}{z^{(2,8,9)}}
+\frac{1746360}{z^{(3,3,13)}}
+\frac{1891890}{z^{(3,4,12)}}\\
&\quad
+\frac{1852200}{z^{(3,5,11)}}
+\frac{1891890}{z^{(3,6,10)}}
+\frac{1869840}{z^{(3,7,9)}}
+\frac{1891890}{z^{(3,8,8)}}
+\frac{2012850}{z^{(4,4,11)}}
+\frac{1995840}{z^{(4,5,10)}}
+\frac{2016000}{z^{(4,6,9)}}\\
&\quad
+\frac{2012850}{z^{(4,7,8)}}
+\frac{1955520}{z^{(5,5,9)}}
+\frac{1995840}{z^{(5,6,8)}}
+\frac{1968840}{z^{(5,7,7)}}
+\frac{2012850}{z^{(6,6,7)}}+\cdots.
\end{split}&&
\end{flalign*}
\begin{flalign*}
\begin{split}
&G_{(4)}^{1D}(z_1,z_2,z_3,z_4)=
\frac{6}{z^{(2,2,2,4)}}
+\frac{8}{z^{(2,2,3,3)}}
+\frac{60}{z^{(2,2,2,6)}}
+\frac{72}{z^{(2,2,3,5)}}
+\frac{72}{z^{(2,2,4,4)}}
+\frac{72}{z^{(2,3,3,4)}}\\
&\quad
+\frac{48}{z^{(3,3,3,3)}}
+\frac{630}{z^{(2,2,2,8)}}
+\frac{720}{z^{(2,2,3,7)}}
+\frac{750}{z^{(2,2,4,6)}}
+\frac{768}{z^{(2,2,5,5)}}
+\frac{720}{z^{(2,3,3,6)}}
+\frac{768}{z^{(2,3,4,5)}}\\
&\quad
+\frac{810}{z^{(2,4,4,4)}}
+\frac{576}{z^{(3,3,3,5)}}
+\frac{720}{z^{(3,3,4,4)}}
+\frac{7560}{z^{(2,2,2,10)}}
+\frac{8400}{z^{(2,2,3,9)}}
+\frac{8820}{z^{(2,2,4,8)}}
+\frac{9000}{z^{(2,2,5,7)}}\\
&\quad
+\frac{9000}{z^{(2,2,6,6)}}
+\frac{8400}{z^{(2,3,3,8)}}
+\frac{9000}{z^{(2,3,4,7)}}
+\frac{9000}{z^{(2,3,5,6)}}
+\frac{9540}{z^{(2,4,4,6)}}
+\frac{9504}{z^{(2,4,5,5)}}
+\frac{7200}{z^{(3,3,3,7)}}\\
&\quad
+\frac{8400}{z^{(3,3,4,6)}}
+\frac{7680}{z^{(3,3,5,5)}}
+\frac{9000}{z^{(3,4,4,5)}}
+\frac{9720}{z^{(4,4,4,4)}}
+\frac{103950}{z^{(2,2,2,12)}}
+\frac{113400}{z^{(2,2,3,11)}}
+\frac{119070}{z^{(2,2,4,10)}}\\
&\quad
+\frac{120960}{z^{(2,2,5,9)}}
+\frac{121590}{z^{(2,2,6,8)}}
+\frac{122040}{z^{(2,2,7,7)}}
+\frac{113400}{z^{(2,3,3,10)}}
+\frac{120960}{z^{(2,3,4,9)}}
+\frac{120960}{z^{(2,3,5,8)}}
+\frac{122040}{z^{(2,3,6,7)}}\\
&\quad
+\frac{127890}{z^{(2,4,4,8)}}
+\frac{128160}{z^{(2,4,5,7)}}
+\frac{129150}{z^{(2,4,6,6)}}
+\frac{128160}{z^{(2,5,5,6)}}
+\frac{100800}{z^{(3,3,3,9)}}
+\frac{113400}{z^{(3,3,4,8)}}
+\frac{108000}{z^{(3,3,5,7)}}\\
&\quad
+\frac{113400}{z^{(3,3,6,6)}}
+\frac{122040}{z^{(3,4,4,7)}}
+\frac{120960}{z^{(3,4,5,6)}}
+\frac{114048}{z^{(3,5,5,5)}}
+\frac{130410}{z^{(4,4,4,6)}}
+\frac{128160}{z^{(4,4,5,5)}}
+\cdots.
\end{split}&&
\end{flalign*}
\begin{flalign*}
\begin{split}
&G_{(5)}^{1D}(z_1,z_2,z_3,z_4,z_5)=
\frac{24}{z^{(2,2,2,2,5)}}
+\frac{36}{z^{(2,2,2,3,4)}}
+\frac{48}{z^{(2,2,3,3,3)}}
+\frac{360}{z^{(2,2,2,2,7)}}
+\frac{480}{z^{(2,2,2,3,6)}}\\
&\quad
+\frac{504}{z^{(2,2,2,4,5)}}
+\frac{576}{z^{(2,2,3,3,5)}}
+\frac{576}{z^{(2,2,3,4,4)}}
+\frac{576}{z^{(2,3,3,3,4)}}
+\frac{384}{z^{(3,3,3,3,3)}}
+\frac{5040}{z^{(2,2,2,2,9)}}\\
&\quad
+\frac{6300}{z^{(2,2,2,3,8)}}
+\frac{6660}{z^{(2,2,2,4,7)}}
+\frac{6840}{z^{(2,2,2,5,6)}}
+\frac{7200}{z^{(2,2,3,3,7)}}
+\frac{7500}{z^{(2,2,3,4,6)}}
+\frac{7680}{z^{(2,2,3,5,5)}}\\
&\quad
+\frac{7848}{z^{(2,2,4,4,5)}}
+\frac{7200}{z^{(2,3,3,3,6)}}
+\frac{7680}{z^{(2,3,3,4,5)}}
+\frac{8100}{z^{(2,3,4,4,4)}}
+\frac{5760}{z^{(3,3,3,3,5)}}
+\frac{7200}{z^{(3,3,3,4,4)}}\\
&\quad
+\frac{75600}{z^{(2,2,2,2,11)}}
+\frac{90720}{z^{(2,2,2,3,10)}}
+\frac{95760}{z^{(2,2,2,4,9)}}
+\frac{98280}{z^{(2,2,2,5,8)}}
+\frac{99000}{z^{(2,2,2,6,7)}}
+\frac{100800}{z^{(2,2,3,3,9)}}\\
&\quad
+\frac{105840}{z^{(2,2,3,4,8)}}
+\frac{108000}{z^{(2,2,3,5,7)}}
+\frac{108000}{z^{(2,2,3,6,6)}}
+\frac{111240}{z^{(2,2,4,4,7)}}
+\frac{112680}{z^{(2,2,4,5,6)}}
+\frac{114048}{z^{(2,2,5,5,5)}}\\
&\quad
+\frac{100800}{z^{(2,3,3,3,8)}}
+\frac{108000}{z^{(2,3,3,4,7)}}
+\frac{108000}{z^{(2,3,3,5,6)}}
+\frac{114480}{z^{(2,3,4,4,6)}}
+\frac{114048}{z^{(2,3,4,5,5)}}
+\frac{119880}{z^{(2,4,4,4,5)}}\\
&\quad
+\frac{86400}{z^{(3,3,3,3,7)}}
+\frac{100800}{z^{(3,3,3,4,6)}}
+\frac{92160}{z^{(3,3,3,5,5)}}
+\frac{108000}{z^{(3,3,4,4,5)}}
+\frac{116640}{z^{(3,4,4,4,4)}}
+\cdots.
\end{split}&&
\end{flalign*}
\begin{flalign*}
\begin{split}
&G_{(6)}^{1D}(z_1,z_2,z_3,z_4,z_5,z_6)=
\frac{120}{z^{(2,2,2,2,2,6)}}
+\frac{192}{z^{(2,2,2,2,3,5)}}
+\frac{216}{z^{(2,2,2,2,4,4)}}
+\frac{288}{z^{(2,2,2,3,3,4)}}\\
&\quad
+\frac{384}{z^{(2,2,3,3,3,3)}}
+\frac{2520}{z^{(2,2,2,2,2,8)}}
+\frac{3600}{z^{(2,2,2,2,3,7)}}
+\frac{3960}{z^{(2,2,2,2,4,6)}}
+\frac{4032}{z^{(2,2,2,2,5,5)}}
+\frac{4800}{z^{(2,2,2,3,3,6)}}\\
&\quad
+\frac{5040}{z^{(2,2,2,3,4,5)}}
+\frac{5184}{z^{(2,2,2,4,4,4)}}
+\frac{5760}{z^{(2,2,3,3,3,5)}}
+\frac{5760}{z^{(2,2,3,3,4,4)}}
+\frac{5760}{z^{(2,3,3,3,3,4)}}
+\frac{3840}{z^{(3,3,3,3,3,3)}}\\
&\quad
+\frac{45360}{z^{(2,2,2,2,2,10)}}
+\frac{60480}{z^{(2,2,2,2,3,9)}}
+\frac{65520}{z^{(2,2,2,2,4,8)}}
+\frac{67680}{z^{(2,2,2,2,5,7)}}
+\frac{68400}{z^{(2,2,2,2,6,6)}}
+\frac{75600}{z^{(2,2,2,3,3,8)}}\\
&\quad
+\frac{79920}{z^{(2,2,2,3,4,7)}}
+\frac{82080}{z^{(2,2,2,3,5,6)}}
+\frac{84240}{z^{(2,2,2,4,4,6)}}
+\frac{85824}{z^{(2,2,2,4,5,5)}}
+\frac{86400}{z^{(2,2,3,3,3,7)}}
+\frac{90000}{z^{(2,2,3,3,4,6)}}\\
&\quad
+\frac{92160}{z^{(2,2,3,3,5,5)}}
+\frac{94176}{z^{(2,2,3,4,4,5)}}
+\frac{97200}{z^{(2,2,4,4,4,4)}}
+\frac{86400}{z^{(2,3,3,3,3,6)}}
+\frac{92160}{z^{(2,3,3,3,4,5)}}
+\frac{97200}{z^{(2,3,3,4,4,4)}}\\
&\quad
+\frac{69120}{z^{(3,3,3,3,3,5)}}
+\frac{86400}{z^{(3,3,3,3,4,4)}}
+\cdots.
\end{split}&&
\end{flalign*}

\normalsize

\subsection{Numerical data of $\chi(\Mbar_{g,n})$}
\label{sec-app-data-chi}

Now the orbifold Euler characteristics of $\Mbar_{g,n}$ can be computed using either methods in
our earlier work \cite{wz2} or the formulas developed this work.
In this subsection let us present some data here.

\vspace{0.3cm}

\small
\renewcommand\arraystretch{1.75}
\begin{tabular}{l|cccc}
\diagbox{$g$}{$n$} &$ 0 $&$ 1 $&$ 2 $&$ 3 $\\
\hline
$0 $ &   &   &   &$ 1$\\
$1 $ & &$ \frac{5}{12} $&$ \frac{1}{2} $&$ \frac{17}{12} $ \\
$2 $ &$ \frac{119}{1440 } $&$  \frac{247}{1440} $&$  \frac{413}{720} $&
$ \frac{89}{32} $ \\
$3 $ &$ \frac{8027}{181440}  $&$  \frac{13159}{72576} $&$  \frac{179651}{181440} $&$ \frac{495611}{72576}$ \\
$4 $ &$ \frac{2097827}{43545600}  $&$  \frac{5160601}{17418240} $&$  \frac{97471547}{43545600} $&$ \frac{1747463783}{87091200}$  \\
$5 $ &$ \frac{150427667}{1916006400}  $&$  \frac{1060344499}{1642291200} $&$  \frac{35763130021}{5748019200} $&$ \frac{157928041517}{2299207680}$  \\
$6 $ &$ \frac{31966432414753}{188305108992000}  $&$  \frac{43927799939987}{25107347865600} $&$  \frac{350875518979697}{17118646272000} $&$ \frac{14466239894532961}{53801459712000}$  \\
$7 $ &$ \frac{21067150021261}{46115536896000}  $&$  \frac{25578458051299001}{4519322615808000 }$&$  \frac{5346168720992921}{68474585088000} $&$ \frac{766050649843508339}{645617516544000}$  \\
$8$ & \tiny{$\frac{27108194937436478387}{18438836272496640000}$} & \tiny{$\frac{71323310082487963309}{3352515685908480000}$} &
\tiny{$\frac{6227476659153540516409}{18438836272496640000}$} &
\tiny{$\frac{409876415908263532817}{70243185799987200}$} \\
\hline
\end{tabular}

\vspace{0.5cm}

\begin{tabular}{l|ccc}
\diagbox{$g$}{$n$} &$ 4 $&$ 5 $ &$6$\\
\hline
$0 $ &$ 2$&$ 7$ &$ 34 $\\
$1 $ &$ \frac{35}{6} $&$ \frac{389}{12} $& $ \frac{1349}{6} $\\
$2 $ &$ \frac{12431}{720} $&$ \frac{189443}{1440} $&$ \frac{853541}{720} $\\
$3 $ &$ \frac{684641}{12096}$&$ \frac{199014019}{362880}$& $ \frac{1103123803}{181440} $\\
$4 $ &$ \frac{9056350741}{43545600}$&$ \frac{71024755987}{29030400}$& $ \frac{1402182822991}{43545600} $\\
$5 $ &$ \frac{701735503159}{821145600}$&$ \frac{135972856739213}{11496038400}$& $\frac{115110462356893}{638668800}$\\
$6 $ &$ \frac{105018494553645499}{26900729856000}$&$ \frac{4680800827073885069}{75322043596800}$ &$\frac{15587244161672916947}{14485008384000}$\\
$7 $ &$ \frac{44501877704266668461}{2259661307904000}$&$ \frac{1601797289485334976137}{4519322615808000}$ &$\frac{3106681102072897118941}{451932261580800}$\\
$8$ & \tiny{$\frac{182685625436225237071349}{1676257842954240000}$}
&  \tiny{$\frac{80688405819001411538000371}{36877672544993280000}$}
 &  \tiny{$\frac{866365437544472661827562757}{18438836272496640000}$}\\
\hline
\end{tabular}

\begin{Remark}
In particular,
the numbers $\chi(\Mbar_{0,n})$ for $n\geq 3$ are all integers:
\ben
1,\quad
2,\quad
7,\quad
34,\quad
213,\quad
1630,\quad
14747,\quad
153946,\quad
\cdots.
\een
This sequence is A074059 on Sloane¡¯s on-line Encyclopedia of Integer Sequences \cite{sl},
which describes a particular weighted counting of stable trees.
Moreover,
there is a refinement of this sequence of integers (A075856)
which counts stable trees with a fixed number of internal edges,
and they are the coefficients (up to an additional factor $\pm n!$)
of the refined orbifold Euler characteristics of $\Mbar_{0,n}$ introduced in \cite[\S 3.1]{wz2}.
See \cite[\S 3]{wz2} and the references within for an introduction of this refined integer sequence.
In particular,
they are special values of the Ramanujan psi polynomials \cite{be, bew},
see \cite[\S 3.5]{wz2} for details.
\end{Remark}

\end{appendices}

\end{document}